\documentclass[10pt]{article}
\usepackage[top=0.9in,bottom=1.1in,left=1.05in,right=1.05in]{geometry}
\usepackage{amsmath,amssymb}
\usepackage{amsthm}
\usepackage{latexsym}
\usepackage{mathrsfs,dsfont}
\usepackage{cancel}
\usepackage{comment}
\usepackage{float}
\usepackage{graphicx}
\usepackage{epstopdf}
\usepackage{color}
\usepackage{enumerate}
\usepackage{natbib}
\DeclareGraphicsExtensions{.eps}
\numberwithin{equation}{section}
\DeclareMathOperator*{\esssup}{ess\,sup}
\newcommand{\MCG}{\mathcal{G}}
\newcommand{\MCC}{\mathcal{C}}
\newcommand{\MCA}{\mathcal{A}}

\newcommand{\MCF}{\mathcal{F}}
\newcommand{\MCO}{\mathcal{O}}
\newcommand{\MCD}{\mathcal{D}}

\newcommand{\MCK}{\mathcal{K}}
\newcommand{\EE}{\mathbb{E}}
\newcommand{\PP}{\mathbb{P}}
\newcommand{\RR}{\mathbb{R}}
\newcommand{\NN}{\mathbb{N}}

\newcommand{\Ltx}{\mathcal{L}_{t,x}}

\newcommand{\infnorm}[1]{\left\lVert#1\right\rVert_\infty}

\newcommand{\Vz}{V^{\delta}}

\newcommand{\vz}{v^{(0)}}
\newcommand{\vo}{v^{(1)}}

\newcommand{\pz}{{\pi^{(0)}}}

\newcommand{\Vzl}{V^{\pz,\delta}}

\newcommand{\Vzp}{V^{\pi,\delta}}

\newcommand{\pzt}{\widetilde\pi^0}
\newcommand{\pot}{\widetilde\pi^1}

\newcommand{\abs}[1]{\left|#1\right|}
\newcommand{\average}[1]{\left\langle#1\right\rangle}

\newcommand{\mc}[1]{\mathcal{#1}}

\newcommand{\ud}{\,\mathrm{d}}

\newcommand{\Wh}[1]{W^{(H)}_{#1}}
\newcommand{\Zh}[1]{Z^{\delta,H}_{#1}}
\newcommand{\half}{\frac{1}{2}}

\newtheorem{theo}{Theorem}[section]
\newtheorem{lem}[theo]{Lemma}

\newtheorem{rem}[theo]{Remark}
\newtheorem{prop}[theo]{Proposition}
\newtheorem{assump}[theo]{Assumption}
\newtheorem{cor}[theo]{Corollary}

\newtheorem{theo*}{Theorem}[section]
\newtheoremstyle{dotlessS}{}{}{\color{blue}}{}{\color{blue}\bfseries}{}{ }{}
\theoremstyle{dotlessS}
\newcommand{\be}{\begin{equation}}
\newcommand{\en}{\end{equation}}

\begin{document}

\title{\vspace{-50pt} Optimal Portfolio  under Fractional Stochastic Environment}
%\author{Ruimeng Hu\\Department of Statistics and Applied Probability \\ University of California, Santa Barbara}
\author{Jean-Pierre Fouque\thanks{Department of Statistics \& Applied Probability,
 University of California,
        Santa Barbara, CA 93106-3110, {\em fouque@pstat.ucsb.edu}. Work  supported by NSF grant DMS-1409434.}
        \and Ruimeng Hu\thanks{Department of Statistics \& Applied Probability,
 University of California,
        Santa Barbara, CA 93106-3110, {\em hu@pstat.ucsb.edu}.}
        }
\date{\today}
\maketitle

\begin{abstract}
Rough stochastic volatility models have attracted a lot of attention recently, in particular for the linear option pricing problem.  
In this paper, starting with power utilities, we propose to use a \emph{martingale distortion representation} of the optimal value function for the nonlinear asset allocation problem in a (non-Markovian) fractional stochastic environment (for all values of the Hurst index $H \in (0,1)$). We rigorously establish a first order approximation of the optimal value, when the return and volatility of the underlying asset are functions of a stationary slowly varying fractional Ornstein-Uhlenbeck process. We prove that this approximation can be also generated by a fixed zeroth order trading strategy providing an explicit strategy which is asymptotically optimal in all admissible controls. Furthermore, we extend the discussion to general utility functions, and obtain the asymptotic optimality of this fixed strategy in a specific family of admissible strategies.

\end{abstract}


\section{Introduction}\label{sec_intro}

In this paper, we study the Merton problem under a non-Markovian fractional stochastic environment, and we are able to provide an explicit trading strategy which is asymptotically optimal in the case of power utilities and asymptotically optimal in a specific family of general utilities. 

The portfolio optimization problem was first studied in the continuous-time framework by Merton \cite{Me:69, Me:71}, where risky assets are considered following the Black-Scholes-Merton model with constant returns and constant volatilities. Under this setup, Merton provided explicit solutions on how to trade stocks and/or how to consume so as to maximize one's utility, when the utility function is of specific types, for instance, {Constant Relative Risk Aversion} (CRRA). After these seminal papers, the optimal portfolio and consumption problem has been extensively studied in financial markets subject to imperfections. For instance, \cite{CoHu:89} and  \cite{KaLeSh:87} studied the case of incomplete markets;  transaction costs have been considered by \cite{MaCo:76} and a user's guide by \cite{GuMu:13}; investment under portfolio constraints are studied by \cite{GrZh:93},  \cite{CvKa:95} and  \cite{ElTo:08}, just to name a few.

A key factor in the Merton problem is the modeling of underlying assets, and empirical studies suggest that volatility is stochastic. In this direction, we refer the readers to \cite{Za:99} for the case of non-linear local volatility models, \cite{ChVi:05} for the case of a particular Heston-like stochastic volatility model, \cite{LoSi:16} for the case of local-stochastic volatility, and \cite{KrSc:03} for the case of general analysis for semimartingale models, to list a few. 

Most of the work has focused on the Markovian models of the volatility. However, in a recent series of papers, non-Markovian  models seem to better describe the data, especially short-range dependence. In \cite{roughvol}, it is beautifully demonstrated that stochastic volatility driven by a fractional Brownian motion (fBm) with Hurst coefficient $H < \half$, so-called \emph{rough fractional stochastic volatility} (RFSV), fit the observed data quite well. \cite{roughvol-limit} and \cite{roughvol-leverage} showed that RFSV is a natural scaling  limit of a general model of Limit Order Book (LOB) based on Hawkes processes. 

Meanwhile,  multi-scale factor models for risky assets were considered in the portfolio optimization problem in \cite{FoSiZa:13} and \cite{Hu:XX}, where return and volatility are driven by a fast mean-reverting factor and a slowly varying factor. Specifically, \cite{FoSiZa:13} heuristically provided the asymptotic approximation to the value function and the optimal strategy for general utility functions, by analyzing a non-linear Hamilton-Jacobi-Bellman partial differential equation (HJB PDE). 

In this paper, we shall consider both the scales and the non-Markovian structure for modeling the underlying assets. As in \cite{FoHu:16}, and in particular because of the relevance for long-term investments (see \cite{FoSiZa:13} for further discussion of the time scales involved), we only consider one \emph{slowly varying fractional stochastic factor} denoted by $\Zh{t}$ for $0<H<1$. The case with fast mean-reverting fractional environment is treated in \cite{FoHu2:17}, while multi-scale models are studied  in the paper in preparation \cite{Hu:XX}.
As in \cite{GaSo:15}, we model $\Zh{t}$ by a fractional Ornstein-Uhlenbeck (fOU) process, which satisfies the following stochastic differential equation (SDE)
\begin{equation*}
\ud \Zh{t} = -\delta a \Zh{t} \ud t + \delta^H \ud W_t^{(H)},
\end{equation*}
where $\delta$ is a small parameter, and $W_t^{(H)}$ is a fractional Brownian motion with Hurst index $H$. We refer to Section~\ref{sec_fBMfOU} for a brief introduction to fBm and fOU, and  to \cite{MaVa:68, ChKaMa:03, Co:07, BiHuOkZh:08, KaSa:11} for more details. 

Pricing options under such RFSV models is indeed a challenge since the model is non-Markovian and PDE tools are no longer available. However, when the fractional stochastic volatility factor is slowly varying (small $\delta$),  one can obtain a practical approximation using the so-called ``epsilon-martingale decomposition'' method designed  in  \cite{FoPaSi:00} and \cite{FPSepsdecomp}. This has been recently carried out for slowly varying RFSV models in \cite{GaSo:15} where a correction to Black-Scholes formula for fractional SV is obtained. Note that the problem is non-Markovian but remains linear in the case of option pricing.

\medskip
\noindent{\bf Main results.} In this paper, we study the nonlinear terminal utility maximization problem under the RFSV model \eqref{def_Zh}. For power utilities, by a martingale distortion representation, we rigorously obtain an expression for the value process at any time and for all $H \in (0,1)$, as well as an expression for the corresponding optimal portfolio.  In the  regime of small $\delta$,  these expressions take the form of a leading order term plus a first order correction of order $\delta^H$. This is done by expanding the martingale distortion representation around a ``frozen" volatility  at the observed value $\Zh{0}$ at time $t=0$. For $H$ relatively small  (close to 0.1 as demonstrated in \cite{roughvol}), the first order correction of the value process is relatively large, and should also be generated by any good practical strategy. Our result nicely shows that the leading order of the optimal strategy, which is explicit in terms of the underlying asset and the current factor level, therefore  easily implemented, will generate the value function up to order $\delta^H$, that is including the first correction. In other words, the $\delta^H$ term in the expression of the optimal strategy is not needed to give such correction to the value process. However, it is given explicitly and can be easily implemented to improve the strategy by taking into account inter-temporal hedging. For general utility functions, using the epsilon-martingale decomposition method and the properties of the risk tolerance function for the Merton problem with constant coefficient, we obtain an approximation for the portfolio value  corresponding to a given strategy, and, as in \cite{FoHu:16} in the Markovian case, we show that this strategy is asymptotically optimal in a specific class of admissible strategies.

\medskip
\noindent{\bf Organization of the paper.} The rest of the paper is organized as follows. In Section~\ref{sec_power}, we present the martingale distortion transformation under general stochastic volatility models first derived in the Markovian case in \cite{Za:99}, and in non-Markovian settings in \cite{Te:04}. Here the drift and volatility of the underlying asset are driven by a stochastic process which is not required to be Markovian nor a semimartingale. We  also present a  generalization to the multi-asset case. In Section~\ref{sec_apptofSV}, we derive the asymptotic results when the stochastic factor is fractional and slowly varying. The approximation to the value process and optimal portfolio are given in Section~\ref{sec_asymppower} and Section~\ref{sec_asymppi} respectively. It is also shown that the leading order of the optimal portfolio is optimal in the full class of admissible strategies up to $\delta^H$, which is numerically illustrated in Section~\ref{sec_num}. The Merton problem with a general utility function is discussed and asymptotic optimality results are presented in Section~\ref{sec_optimality}. We conclude in Section~\ref{sec_conclusion}.

\section{Merton Problem with Power Utilities and Stochastic Environment}\label{sec_power}

Denote by $S_t$ the underlying asset price whose return and volatility are driven by a stochastic factor $Y_t$, 
\begin{align}
\ud S_t = \mu(Y_t) S_t \ud t +  \sigma(Y_t) S_t \ud W_t,\label{def_StunderY}
\end{align}
with assumptions on $\mu(y)$ and $\sigma(y)$ to be specified later.  Here $Y_t$ is a general stochastic process that is adapted to $(\mc{G}_t)$, the natural filtration generated by $\{W_u^Y: u \leq t\}$, and $W_t^Y$ is a Brownian motion generally correlated with the Brownian motion $W_t$ driving the price $S_t$:
$$\ud \average{W, W^Y}_t = \rho \ud t, \quad \abs{\rho }<1.$$
Also define $(\MCF_t)$ as the natural filtration generated by $(W_t, W_t^Y)$.

Denote by  $\pi$ the investor's strategy and by $X_t^\pi$ the corresponding wealth process. The quantity $\pi_t \in \MCF_t$ represents the amount of money invested in the risky asset at time $t$, 
%when the current asset price is $x$ and the fractional stochastic factor's level is at $z$, 
with the remaining held in a money account paying interest at a constant rate $r$. Without loss of generality, we will take $r=0$ throughout.
%while the rest $X_t^\pi - \pi_t$ earns a risk-free interest rate $r$ (constant). 
Assuming that the strategy $\pi$ is self-financing,  the dynamics of the wealth process $X_t^\pi$ is given by:
\begin{equation}
\ud X_t^\pi = \pi_t\mu(Y_t) \ud t + \pi_t\sigma(Y_t) \ud W_t.\label{def_XtunderY}
\end{equation}

The investor's goal is to find the optimal strategy so as to maximize her expected utility of terminal wealth. Mathematically, she aims at identifying the optimal value
\begin{equation}\label{def_Vt}
V_t := \esssup_{\pi \in \MCA_t}\EE\left[U(X_T^\pi)\vert \MCF_t\right],
\end{equation}
and the optimal strategy $\pi^\ast$, given her preference described by a utility function $U(\cdot)$. In this section and Section \ref{sec_apptofSV}, we consider the power utility case:
\begin{equation}\label{def_power}
U(x) = \frac{x^{1-\gamma}}{1-\gamma}, \quad \gamma >0, \quad \gamma \neq 1,
\end{equation}
and the set $\MCA_t$ is the class of all admissible strategies: 
\begin{equation}\label{def_MCA}
\MCA_t := \left\{\pi \text{ is }  (\MCF_t)\text{-adapted}: X_s^\pi \text{ in } \eqref{def_XtunderY} \text{ stays nonnegative } \forall s \geq t, \text{ given } \MCF_t\right\},
\end{equation}
where zero is an absorbing state for $X_t^\pi$ (bankruptcy). Additionally, for the power utility case,  we require  that for all $\pi \in \MCA_t$, the following integrability conditions are satisfied:
	\begin{align}\label{assump_strategies}
	\sup_{t\in[0,T]}\EE\left[ \left(X_t^\pi\right)^{2p(1-\gamma)}\right] < +\infty,\; for \; some  \;p > 1, \quad and \quad \EE \left[\int_0^T\left(X_t^\pi\right)^{-2\gamma}    \pi_t^2  \sigma^2(Y_t)\ud t \right] < \infty.
	\end{align}
Later, in Section~\ref{sec_optimality}, we will discuss the case with general utility functions.

In order to motivate the martingale  distortion transformation that we will introduce in Section~\ref{sec_martdistortion}, we first recall in the next subsection the distortion transformation obtained by \cite{Za:99}  in the Markovian case with power utility \eqref{def_power}. We also stated in Remark \ref{rem_geneMDT} that results can be generalized to the multi-asset case, when the returns and volatilities of stocks are driven by the same randomness $W^Y$.

\subsection{The Distortion Transformation}\label{sec_distorttrans}

In the Markovian setup, $Y_t$ is a diffusion process following the stochastic differential equation of the form
	\begin{equation*}
	\ud Y_t = k(Y_t) \ud t + h(Y_t) \ud W_t^Y,
	\end{equation*}
	and the value function $V(t,x,y) := \sup_{\pi \in \MCA_t}\EE\left[U(X_T^\pi)\vert X
	_t = x, Y_t = y\right]$ is a solution to the Hamilton-Jacobi-Bellman (HJB) equation given in \cite{FoSiZa:13}. The distortion transformation is given by 
	\begin{equation}\label{eq_distorttrans}
	V(t,x,y) = \frac{x^{1-\gamma}}{1-\gamma}\Psi(t,y)^q, 
	\end{equation}
	with 
	\begin{equation}\label{def:q}
	 q = \frac{\gamma}{\gamma + (1-\gamma)\rho^2}
	 \end{equation}
	   which results in canceling $(\Psi_y)^2$ terms in the HJB equation. 
	Consequently, $\Psi$ solves the linear PDE
	\begin{equation*}
	\Psi_t + \left(\half h^2(y)\partial_{yy} + k(y)\partial_y + \frac{1-\gamma}{\gamma}\lambda(y)\rho h(y)\partial_y\right)\Psi + \frac{1-\gamma}{2q\gamma}\lambda^2(y)\Psi = 0, \quad \Psi(T,y) = 1,
	\end{equation*}
	where $\lambda(y)$ is the Sharpe ratio $\lambda(y) := \mu(y)/\sigma(y)$. 
	
	By Feynman-Kac formula, we observe that $\Psi$ can be expressed as
	\begin{equation}\label{eq_Psi}
	\Psi(t,y) =\widetilde \EE\left[\left.e^{\frac{1-\gamma}{2q\gamma}\int_t^T \lambda^2( Y_s)\ud s }\right\vert Y_t = y\right],
	\end{equation}
where under $\widetilde \PP$,  $ \widetilde W_t^Y = W_t^Y - \int_0^t \rho\left(\frac{1-\gamma}{\gamma}\right)\lambda( Y_s) \ud s $ is a standard Brownian motion.
	
	The formula in the next subsection generalizes \eqref{eq_Psi} without using any PDE argument.

\subsection{Martingale Distortion Transformation}\label{sec_martdistortion}

The martingale distortion transformation is motived by the formulas \eqref{eq_distorttrans} and \eqref{eq_Psi}. It has been derived in \cite{Te:04} with a slightly different utility function. For the sake of clarity, we restate it here, and we propose a short proof based on verification using stochastic calculus. We comment that the results and proofs presented below can be extended straightforwardly to the multi-asset case (see Remark \ref{rem_geneMDT}). 
Here and in the rest the paper, we only present  the single asset case for  simplicity of notations.

Note that in the following Proposition~\ref{thm_martdistort}, $(Y_t)$ is a general stochastic process adapted to $(\mc{G}_t)$  \emph{which does not need to be Markovian, nor a semimartingale}. In particular, in  Section~\ref{sec_apptofSV}, we will be able to apply it to the case  $(Y_t)$ being a fractional process.

Let us assume that the Sharpe-ratio $\lambda(\cdot)$ is bounded. Define a new probability measure $\widetilde \PP$ by
\begin{equation}\label{def_Ptilde}
\frac{\ud \widetilde \PP}{\ud \PP} = \exp\left\{-\int_0^T a_s \ud W_s^Y - \half \int_0^T a_s^2 \ud s \right\},
\end{equation}
where $a_t$ is given by
\begin{equation}\label{def:at}
a_t = -\rho\left(\frac{1-\gamma}{\gamma}\right)\lambda(Y_t),
\end{equation}
and therefore, is bounded and $\MCG_t$-adapted.
Then, $\widetilde W_t^Y :=  W_t^Y+ \int_0^t a_s \ud s$ is a standard Brownian motion under $\widetilde \PP$. We now make the following {\it model assumptions}.

\begin{assump}\label{assump_power}
\begin{enumerate}[(i)]
\item The SDE \eqref{def_StunderY} for $S_t$ has a unique strong solution. The function $\lambda(\cdot)$ is assumed to be bounded  and $C^2(\RR)$. The function $\lambda'(\cdot)$ is bounded and $\lambda''(\cdot)$ is at most polynomially growing.

\item Define the $\widetilde \PP$-martingale
\begin{equation}\label{def_Mtmartdistort}
M_t = \widetilde \EE\left[\left.e^{\frac{1-\gamma}{2q\gamma}\int_0^T \lambda^2(Y_s)\ud s} \right\vert \MCG_t\right],
\end{equation}
and write its representation
\begin{equation}\label{def_xi}
 \ud M_t = M_t\xi_t \ud \widetilde W_t^Y.
\end{equation}
%by the Martingale Representation Theorem.
We assume
\begin{equation}\label{assump_xi}
\EE\left[e^{c_\xi\int_0^T \xi_t^2\ud t}\right] < \infty, 
\end{equation}
where the constant $c_\xi$ is given by 
$c_\xi = \frac{16 (1-\gamma)^2\rho^2 p^2q^2}{\gamma^2}$ for  $\gamma < 1$, and $c_\xi = \frac{16 (1-\gamma)^2\rho^2 p^2q^2}{\gamma^2} - \frac{4p(1-\gamma)}{\gamma^2}$ for  $\gamma >1$. The parameter $p$ is introduced in \eqref{assump_strategies} and $q$ is defined by \eqref{def:q}.

\end{enumerate}
\end{assump}

\begin{prop}\label{thm_martdistort}
	Let $S_t$ follow the dynamics \eqref{def_StunderY}, and suppose the objective is \eqref{def_Vt} with power utility function \eqref{def_power}. Under Assumption~\ref{assump_power}, the value process $V_t$ defined in \eqref{def_Vt} is given by
	\begin{equation}\label{def:Vcorrelated}
	V_t=\frac{X_t^{1-\gamma}}{1-\gamma} \left[\widetilde{\EE}\left(\left.e^{\frac{1-\gamma}{2q\gamma}\int_t^T\lambda^2(Y_s)\ud s}\right\vert \mc{G}_t\right)\right]^q.
	\end{equation}
	The expectation
	$\widetilde \EE[\cdot]$ is computed with respect to $\widetilde{\PP}$ introduced in \eqref{def_Ptilde}.
	The parameter $q$ is given in terms of $\gamma$ and $\rho$ by \eqref{def:q}.
	The optimal strategy $\pi^\ast$ is
	\begin{equation}\label{def_pioptimal}
	\pi^\ast_t = \left[\frac{\lambda(Y_t)}{\gamma \sigma(Y_t)} + \frac{\rho q \xi_t}{\gamma \sigma(Y_t)}\right] X_t,
	\end{equation}
where $\xi_t$ is given in \eqref{def_xi}.

	The conditioning with respect to $\MCG_t$ corresponds to the separation of variable in the Markovian case presented in Section~\ref{sec_distorttrans}.
\end{prop}

\begin{rem}\label{rem_martdistort}
	\quad
	
%	(i) With the moving average representation, the volatility model \eqref{eq_Zt} in Section~\ref{sec_problem} satisfies the assumption here. \todo{mentioned later}
\begin{enumerate}[(i)]	
	\item
	Note that $\gamma=1$ in \eqref{def_power} is the log utility case, which can be treated separately.
	
\item
For the degenerate case $\lambda(y) \equiv \lambda_0$, the value process $V_t$ is reduced to 
		\begin{equation*}
		V_t  = \frac{X_t^{1-\gamma}}{1-\gamma}e^{\frac{1-\gamma}{2\gamma}\lambda^2_0(T-t)}.
		\end{equation*}
	The quantity $a_t =-\rho\left( \frac{1-\gamma}{\gamma}\right)\lambda_0$ is a constant and a direct computation from \eqref{def_Mtmartdistort} yields $\xi_t = 0$. Consequently, the optimal control $\pi^\ast$ becomes
	\begin{equation*}
 \pi_t^\ast = \frac{\lambda_0}{\gamma\sigma(Y_t)}X_t.
	\end{equation*}
	In this case, both $V_t$ and $\pi^\ast_t$ do not depend on $a_t$ and $q$ as expected.
	
\item
 In the uncorrelated case $\rho=0$, the problem is already ``linear'',  since $q = 1$.  The value process $V_t$ and the optimal control $\pi^\ast$ are simplified as
	\begin{equation*}
	V_t = \frac{X_t^{1-\gamma}}{1-\gamma}\EE\left[\left.e^{\frac{1-\gamma}{2\gamma}\int_t^T \lambda^2(Y_s) \ud s}\right\vert \mc{G}_t\right], \quad \pi_t^\ast = \frac{\lambda(Y_t)}{\gamma\sigma(Y_t)} X_t.
	\end{equation*} 
\end{enumerate}
\end{rem}

\begin{proof}[Proof of Proposition~\ref{thm_martdistort}]
	The proof follows a verification argument, that is, in order to prove that $V_t$ is indeed the value process and $\pi^\ast$ given in \eqref{def_pioptimal} is optimal, one needs to prove (i) for any control $\pi_t \in \MCA_t$, the process \eqref{def:Vcorrelated} is a supermartingale, and (ii) $V_t$ is a martingale under the control \eqref{def_pioptimal} which needs to be admissible.
	
	Let $\alpha_t$ be the proportion of the wealth invested in $S_t$ at time $t$, namely, $\pi_t = \alpha_t X_t$, then the wealth process \eqref{def_XtunderY} can be rewritten as:
	\begin{equation}
	\ud X_t= X_t\left[\alpha_t\mu(Y_t) \ud t+\alpha_t\sigma(Y_t)\ud W_t\right].
	\end{equation}
	In the following proof, we shall first derive the drift part of $\ud V_t$, then obtain
	$\alpha_t^\ast$ by maximizing the drift over $\alpha$, and finally show that the drift part evaluating at  $\alpha_t^\ast$ is zero with the right choice of $a_t$ and $q$.

	Recall the $\widetilde\PP$-martingale $M_t$ defined in \eqref{def_Mtmartdistort}, and rewrite $V_t$ using $M_t$ as
	\begin{equation}\label{eq_Vt}
	V_t=\frac{X_t^{1-\gamma}}{1-\gamma}e^{N_t}M_t^q,
	\end{equation}
	where $N_t=-\frac{1-\gamma}{2\gamma}\int_0^t\lambda^2(Y_s)\ud s$.
	In the following derivation, we use the short notation $\lambda=\lambda(Y_t), \mu=\mu(Y_t), \sigma=\sigma(Y_t)$. By It\^o's formula applied to $V_t$ in \eqref{eq_Vt}, we deduce
	\begin{align*}
	\ud V_t=&\left(X_t^{-\gamma}\ud X_t-\frac{\gamma}{2}X_t^{-\gamma-1}\ud\average{X}_t\right)e^{N_t}M_t^q +\frac{X_t^{1-\gamma}}{1-\gamma}e^{N_t} M_t^q \ud N_t 
	+\frac{X_t^{1-\gamma}}{1-\gamma}e^{N_t}qM_t^{q-1}\ud M_t\\
	&+\frac{1}{2}\frac{X_t^{1-\gamma}}{1-\gamma}e^{N_t}q(q-1)M_t^{q-2}\ud \average{M}_t
	+d\average{\frac{X^{1-\gamma}}{1-\gamma}e^{N},M^q}_t\\
	=&
	\left(X_t^{-\gamma}X_t\alpha_t\mu-\frac{\gamma}{2}X_t^{-\gamma-1}X_t^2\alpha_t^2\sigma^2\right)e^{N_t}M_t^q\ud t +\frac{X_t^{1-\gamma}}{1-\gamma}e^{N_t}M_t^q\left(-\frac{1-\gamma}{2\gamma}\lambda^2\right)\ud t \\
	&+	\frac{X_t^{1-\gamma}}{1-\gamma}e^{N_t}qM_t^{q-1}M_t\xi_ta_t \ud t
	+ \frac{1}{2}\frac{X_t^{1-\gamma}}{1-\gamma}e^{N_t}q(q-1)M_t^{q-2}M_t^2\xi_t^2\ud t + X_t^{-\gamma}e^{N_t}qM_t^{q-1}\rho X_t\alpha_t\sigma M_t\xi_t\ud t \\
	 &+ \frac{X_t^{1-\gamma}}{1-\gamma}e^{N_t}M_t^q \left[ (1-\gamma) \alpha_t\sigma \ud W_t + q\xi_t \ud W_t^Y\right].
	%\,\mbox{Martingale}.
	\end{align*}
	and we claim that the last term is a true martingale for any admissible strategy $\pi \in \MCA_t$. This follows from the boundedness of $e^{N_t} M_t^q$ guaranteed by the boundedness of $\lambda(\cdot)$, and square integrability of $X_t^{1-\gamma}\alpha_t\sigma$ and $X_t^{1-\gamma}\xi_t$. More precisely, one has:
	\begin{align*}
	&\EE\left[\int_0^T \left(X_t^\pi\right)^{2-2\gamma}\alpha_t^2 \sigma^2(Y_t) \ud t \right] = \EE\left[\int_0^T \left(X_t^\pi\right)^{-2\gamma}\pi_t^2 \sigma^2(Y_t) \ud t \right] < \infty,\\
	&\EE\left[\int_0^T \left(X_t^\pi\right)^{2-2\gamma}\xi_t^2 \ud t \right] \leq 
	\left[\EE\int_0^T \left(X_t^\pi\right)^{2p(1-\gamma)}\ud t \right]^\frac{1}{p}\left[\EE\int_0^T \xi_t^{2p/(p-1)} \ud t \right]^\frac{p-1}{p} < \infty.
	\end{align*}
	by the admissibility \eqref{assump_strategies} of $\pi$  and Assumption \ref{assump_power}(ii) which implies finite moments of $\xi_t$.
	
	By rewriting $\ud V_t=X_t^{1-\gamma}e^{N_t}M_t^qD_t(\alpha_t)\ud t+\ud\,\mbox{Martingale}$, the drift factor $D_t(\alpha_t)$  takes the form:
	\begin{align*}
	D_t(\alpha_t):=
	\alpha_t\mu-\frac{\gamma}{2}\alpha_t^2\sigma^2-\frac{\lambda^2}{2\gamma}
	+\frac{q}{1-\gamma}a_t\xi_t+ \frac{q(q-1)}{2(1-\gamma)}\xi_t^2 + 
	\rho q\alpha_t\sigma\xi_t.
	\end{align*}
	Differentiating $D_t(\alpha_t)$ with respect to $\alpha$ and checking the second order condition, one obtains the maximizer
	\begin{equation}\label{eq_alphaoptimal}
	\alpha_t^\ast=\frac{\mu}{\gamma\sigma^2}+\frac{\rho q\xi_t}{\gamma\sigma} = \frac{\lambda}{\gamma\sigma} + \frac{\rho q \xi_t}{\gamma\sigma}.
	\end{equation}
	Evaluating the drift factor $D_t$ at $\alpha_t^\ast$ produces
	%Injecting this maximizer into $D_t(\alpha_t)$ and using $\frac{\mu^2}{\gamma\sigma^2}-\frac{\gamma \sigma^2\mu^2}{2\gamma^2\sigma^4}-\frac{\lambda^2}{2\gamma}=0$ thanks to $\lambda=\mu/\sigma$, 
	\begin{align}
	D_t(\alpha_t^\ast)= q\xi_t\left(\frac{a_t }{1-\gamma}+ \frac{\lambda \rho }{\gamma}\right) + \frac{q\xi_t^2}{2}\left[\frac{\rho^2q}{\gamma}+\frac{q-1}{1-\gamma}\right].\label{eq_Dtoptimal}
	\end{align}
	Then, the drift factor $D_t(\alpha_t^\ast)$ vanishes under the choices \eqref{def:q} for $q$ and \eqref{def:at} for $a_t$. Note that the other choice $\xi = 0$ would only lead to the degenerate case $\lambda(\cdot)$ constant considered in Remark \ref{rem_martdistort}(ii). Otherwise,  since $\xi_t$ does not depend on $a_t$, \eqref{def:q} and \eqref{def:at} is the only choice to zero out $D_t(\alpha_t^\ast)$. 	Also note that with the choice \eqref{def:q} for $q$, the term $\xi_t^2$ is canceled which corresponds to the cancellation of the nonlinear term $(\partial_y \Phi)^2$ in the PDE argument reviewed in Section~\ref{sec_distorttrans}.
	
	In addition, using the relation $\pi_t = \alpha_tX_t$ and equation \eqref{eq_alphaoptimal} for $\alpha^\ast_t$, the wealth process following $\pi_t^\ast$ solves the SDE
	\begin{equation*}
	\ud X_t^{\pi^\ast} = X_t^{\pi^\ast}\left[\frac{\lambda^2(Y_t) + \rho q \lambda(Y_t)\xi_t}{\gamma} \ud t+\frac{\lambda(Y_t) + \rho q \xi_t}{\gamma}\ud W_t\right],
	\end{equation*}
	thus, it stays nonnegative, which implies that $\pi^\ast_t = \alpha_t^\ast X_t$ satisfies \eqref{def_MCA}. In order to check the condition \eqref{assump_strategies}, we first notice that
		\begin{align*}
		\EE\left[\int_0^T \left(X_t^{\pi^\ast}\right)^{-2\gamma}(\pi_t^\ast)^2 \sigma^2(Y_t) \ud t \right] &= 	\EE\left[\int_0^T \left(X_t^{\pi^\ast}\right)^{2-2\gamma}\left(\frac{\lambda(Y_t)}{\gamma} + \frac{\rho q \xi_t}{\gamma}\right)^2 \ud t \right].
		\end{align*}
Then, by H\"{o}lder inequality, the boundedness of $\lambda$ and integrability condition of $\xi_t$, it suffices to verify $\sup_{t\in[0,T]}\EE\left[ \left(X_t^{\pi^\ast}\right)^{2p(1-\gamma)}\right] < +\infty$, for some  $p > 1$. To this end, we compute
\begin{align*}
\EE\left[ \left(X_t^{\pi^\ast}\right)^{2p(1-\gamma)}\right]  &= \EE\left[e^{\int_0^t \frac{2p(1-\gamma)}{\gamma}(\lambda^2 + \rho q \lambda\xi_s) - \frac{p(1-\gamma)}{\gamma^2}(\lambda + \rho q \xi_s)^2\ud s + \int_0^t \frac{2p(1-\gamma)}{\gamma}(\lambda + \rho q \xi_s)\ud W_s}\right] \\
&\leq \EE\left[e^{\int_0^t \frac{4p(1-\gamma)}{\gamma}(\lambda^2 + \rho q \lambda\xi_s) \ud s + \int_0^t \left(\frac{ 8p^2(1-\gamma)^2}{\gamma^2}-\frac{2p(1-\gamma)}{\gamma^2}\right)(\lambda + \rho q \xi_s)^2\ud s}\right]\\
& \quad \times
\EE\left[e^{-\int_0^t  \frac{8p^2(1-\gamma)^2}{\gamma^2}(\lambda + \rho q \xi_s)^2\ud s + \int_0^t \frac{4p(1-\gamma)}{\gamma}(\lambda + \rho q \xi_s)\ud W_s}\right].
\end{align*}
By exponential moments of $\xi^2$ in \eqref{assump_xi} under Assumption \ref{assump_power}(ii) and the boundedness of $\lambda$, the first expectation is finite uniformly in $t \in [0,T]$, while the second expectation is one by Novikov's condition. Thus we have obtained the desired results.

\end{proof}

\begin{rem}
Our assumption on $(\xi_t)_{t\in[0,T]}$ is stronger than in other papers (see \cite{Te:04, NaZa:13}), but it allows us to fully justify that $\pi^\ast$ given in \eqref{def_pioptimal} is admissible. In Section \ref{sec_apptofSV}, we will see that this assumption is satisfied for our fractional stochastic environment model.
\end{rem}

Even though the results in the next section are presented only in the single asset case for simplicity, we  state here the  formula
for the multiple asset case. The derivation is a tedious exercise.
%here to show it works in the multiple setting.
%Results in Proposition \ref{thm_martdistort} can be generalized to the mutiple asset case. We state the setup and results in the following remark without proof. 

\begin{rem}[Generalization to the multi-asset case]\label{rem_geneMDT} \quad

Let $\mathbf{S}_t := [S_t^1, S_t^2, \ldots, S_t^n]$ be n risky assets modeled by
\begin{align*}%\label{def_MultiSt}
\ud S_t^i = \mu^i(Y_t^i)S_t^i \ud t + \sum_{j=1}^n\sigma_{ij}(Y_t^i)S_t^i\ud W_t^j, \quad i = 1, 2, \ldots n.
\end{align*}
Here, each $S_t^i$ is driven by its own stochastic factor $Y_t^i$, but all factors are adapted to the same single Brownian motion $W_t^Y$ with the correlation structure:
\begin{equation*}
\ud \average{W^i, W^j}_t = 0,  \quad \ud \average{W^i, W^Y}_t = \rho \ud t, \quad \forall\, i, j = 1, 2, \ldots, n, \quad n\rho^2<1.
\end{equation*}
Denote by $\pi = \left[\pi^1, \pi^2,\cdots, \pi^n \right]^\dagger \in \MCF_t$ the trading vector such that $\pi^i_t$ represents the amount of money invested into $S_t^i$ at time $t$ ($\dagger$ denotes the matrix transpose). In this multi-asset setup, under self-financing assumption and $r=0$, the wealth process $X_t$ satisfies
\begin{align*}
\ud X_t &= 
%\sum_{i=1}^n \pi_t^i \mu^i(Y_t^i) \ud t + \sum_{i=1}^n \pi_t^i \sum_{j=1}^k \sigma_{ij}(Y_t^i)\ud W_t^j = 
 \pi_t \cdot \mu(\mathbf{Y}_t) \ud t + \pi_t \cdot \sigma(\mathbf{Y}_t) \ud \mathbf{W}_t,
\end{align*}
with vector notations $\mathbf{Y}_t := [Y_t^1, Y_t^2, \ldots, Y_t^n]^\dagger$, $\mu(\mathbf{Y}_t) := [\mu^1(Y_t^1), \mu^2(Y_t^2), \cdots, \mu^n(Y_t^n)]^\dagger$, $\sigma(\mathbf{Y}_t) := \left.\sigma_{i,j}(Y_t^i)\right.$ as a square matrix of size n, and $\mathbf{W}_t := [W_t^1, W_t^2, \cdots, W_t^n]^\dagger$.

Assume $\Sigma(\mathbf{Y}_t) = \sigma(\mathbf{Y}_t)\sigma(\mathbf{Y}_t)^\dagger$ is invertible and positive definite, the function $\Lambda(\mathbf{Y}_t) = \sigma(\mathbf{Y}_t)^{-1}\mu(\mathbf{Y}_t)$ is bounded with bounded derivatives, and define the probability measure $\widetilde \PP$ as in \eqref{def_Ptilde} with
$a_t = -\rho\left(\frac{1-\gamma}{\gamma}\right) \mathds{1}_n^\dagger \sigma(\mathbf{Y}_t)^{-1}\mu(\mathbf{Y}_t)$.
Define the $\widetilde \PP$-martingale 
\begin{equation*}
M_t=\widetilde \EE\left(\left.e^{\frac{1-\gamma}{2q\gamma}\int_0^T\mu(\mathbf{Y}_s)^\dagger \Sigma (\mathbf{Y}_s)^{-1}\mu(\mathbf{Y}_s)\ud s}\right\vert \mc{G}_t\right),
\end{equation*}
and assume $\xi_t$ given by the Martingale Representation Theorem  satisfies \eqref{assump_xi}, then the portfolio value $V_t$ can be expressed as
	\begin{equation*}
	V_t=\frac{X_t^{1-\gamma}}{1-\gamma} \left[\widetilde{\EE}\left(\left.e^{\frac{1-\gamma}{2q\gamma}\int_t^T \mu(\mathbf{Y}_s)^\dagger \Sigma (\mathbf{Y}_s)^{-1}\mu(\mathbf{Y}_s)\ud s}\right\vert \mc{G}_t\right)\right]^q, 
	\end{equation*}
	where $\widetilde \EE$ is calculated under $\widetilde \PP$ and 
	$q$ is constant chosen to be:
	\begin{equation*}
	 q=\frac{\gamma}{\gamma+(1-\gamma)\rho^2n}.
	\end{equation*}
	The optimal control $\pi^\ast$ is given by
	\begin{equation*}
	\pi^\ast_t = \left[\frac{\Sigma(\mathbf{Y}_t)^{-1}\mu(\mathbf{Y}_t)}{\gamma} + \frac{\rho q \xi_t\sigma^{-1}(\mathbf{Y}_t)^\dagger\mathds{1}_n }{\gamma}\right] X_t,
	\end{equation*}
with $\mathds{1}_n$ being a n-vector of ones. 
	\end{rem}

\section{Application to Fractional Stochastic Environment }\label{sec_apptofSV}

In this section, we first briefly review the fractional Brownian motion (fBm) and fractional Ornstein-Uhlenbeck (fOU) processes, and then introduce the slowly varying fOU process. Under such a model, we will derive an approximation of the portfolio value $V_t$ based on results in Proposition~\ref{thm_martdistort}. More importantly, note that the optimal trading strategy $\pi^\ast$ given by \eqref{def_pioptimal} is not explicit due to the presence of $\xi_t$ given by the martingale representation theorem, and we will obtain an explicit approximation to this optimal strategy.

\subsection{Fractional Brownian Motion and Fractional Ornstein-Uhlenbeck Processes}\label{sec_fBMfOU}

A fractional Brownian motion is a continuous Gaussian process $(\Wh{t}
)$ with zero mean and the covariance structure:
\begin{equation}
\EE\left[\Wh{t}\Wh{s}\right] = \frac{\sigma_H^2}{2}\left(\abs{t}^{2H} + \abs{s}^{2H} - \abs{t-s}^{2H}\right),
\end{equation}
where $\sigma_H$ is a positive constant and $H \in (0,1)$ is called Hurst index. According to \cite{MaVa:68}, $\Wh{t}$ has the following moving-average stochastic integral representation:
\begin{equation}\label{eq_movingfBM}
\Wh{t} = \frac{1}{\Gamma(H+\half)} \int_{\RR} \left((t-s)_+^{H-\half} - (-s)_+^{H-\half}\right) \ud W_s,
\end{equation}
where $(W_t)_{t\in\RR^+}$ is the usual Brownian motion and $(W_t)_{t\in\RR^-} := \left(B_{-t}\right)_{t\in\RR^-}$ is another Brownian motion independent of $(W_t)_{t\in\RR^+}$. With \eqref{eq_movingfBM}, $\sigma_H^2$ is calculated as $\sigma^2_H = ({\Gamma(2H+1)\sin(\pi H)})^{-1}$.

Now we consider the Langevin equation with fractional Brownian motion
\begin{equation}\label{eq_fOU}
\ud Z_t^H = -aZ_t^H \ud t + \ud \Wh{t},
\end{equation}
with the initial condition $Z_0^H = \eta$. In \cite{ChKaMa:03}, it is proved that
\begin{equation*}
Z_t^{H,\eta} := e^{-at}\left(\eta+ \int_0^t e^{au}\ud \Wh{u}\right)
\end{equation*}
is the unique almost surely continuous process that solves equation \eqref{eq_fOU}, where $\int_0^t e^{au}\ud \Wh{u}$ exists as a path-wise Riemann-Stieltjes integral (by integration by parts) and is almost surely continuous in $t$. Particularly, for $t\in\RR^+$,
\begin{equation}\label{eq_fOUsol}
Z_t^H := \int_{-\infty}^t e^{-a(t-s)} \ud \Wh{s} = \Wh{t} - a\int_{-\infty}^t e^{-a(t-s)}\Wh s \ud s,
\end{equation}
is a stationary solution with initial condition $\eta = Z_0^H$, and every other stationary solution has the same distribution as $Z_t^H$. In the sequel, we shall only consider this stationary solution
% \eqref{eq_fOUsol} and with a slight abuse of notation, we denote it by $Z_t^H$ 
and call it the \emph{stationary fractional Ornstein-Uhlenbeck process}. 

It has zero mean and (co)variance structure:
\begin{align}\label{eq_Zhvar}
\sigma_{ou}^2 = \half a^{-2H}\Gamma(2H+1)\sigma_H^2, \quad \EE\left[Z_t^H Z_{t+s}^H\right] =  \sigma^2_{ou}\MCC_Z(s),
\end{align}
where $\MCC_Z(s)$ is given by
\begin{equation}
\MCC_Z(s) =  \frac{2\sin(\pi H)}{\pi} \int_0^\infty \cos(asx)\frac{x^{1-2H}}{1+x^2}\ud x.
\end{equation}
Using the moving-average representation \eqref{eq_movingfBM} for $\Wh{t}$, the stationary solution \eqref{eq_fOUsol} can be expressed as:
\begin{equation}\label{eq_fOUker}
Z_t^H = \int_{-\infty}^t \mc{K}(t-s)\ud W_s^Z,
\end{equation}
where $\left(W_t^Z\right)_{t\in\RR}$ is a standard BM on $\RR$ as described in \eqref{eq_movingfBM}, with the superscript $Z$ indicating that it drives the process $Z_t^H$. 
The kernel $\mc{K}$ is defined by
\begin{equation}\label{def_kernel}
\mc{K}(t) = \frac{1}{\Gamma(H+\half)} \left[t^{H-\half} - a \int_0^t (t-s)^{H-\half}e^{-as}\ud s\right].
\end{equation}
We refer to \cite[Section 2.2]{GaSo:15} for  asymptotic properties of $\MCK(t)$ when $t \ll 1$ and $t \gg 1$, for short-range correlation properties when $H \in (0,\half)$, and for long-range correlation properties when $H \in (\half,1)$. In what follows, we will be mainly interested in the case $H < \half$ as explained in the introduction, but our asymptotic results are also valid for $H > \half$. As noted in \cite[Appendix B]{GaSo:15}, a more general class of Gaussian volatility factors can be considered. But for the sake of simplicity and length ,we restrict ourselves to the case of fOU process.

\subsection{The Slowly Varying fOU Process}
As explained in the introduction, we consider the slowly varying fractional factor denoted by $\Zh{t}$. In the regime of small $\delta$, $\Zh{t}$ is defined as a rescaled stationary fOU process,
\begin{equation}\label{def_Zh}
\Zh{t} = \delta^H \int_{-\infty}^t e^{-\delta a(t-s)} \ud \Wh{s} = \int_{-\infty}^t \MCK^\delta(t-s) \ud W_s^Z, \quad \mc{K}^\delta(t) = \sqrt\delta \mc{K}(\delta t),
\end{equation}
where $\Wh{t}$ is a fBm driven by the Brownian motion $W_t^Z$ via \eqref{eq_movingfBM}, and $\MCK(t)$ is given in \eqref{def_kernel}. According to Section~\ref{sec_fBMfOU}, $\Zh{t}$
%\begin{equation}
%\Zh{t} = \delta^H \int_{-\infty}^t e^{-\delta a(t-s)} \ud \Wh{s} = \delta^H W_t^H - \delta^{1+H}a\int_{-\infty}^t e^{-\delta a(t-s)} W_s^H \ud s,
%\end{equation}
% According to \eqref{eq_fOUsol} and \eqref{eq_fOUker}, $\Zh{t}$ also has the moving-average 
%stochastic representation
%\begin{equation}\label{eq_Zhker}
%\Zh{t} = \int_{-\infty}^{t} \mc{K}^\delta(t-s) \ud W_s^Z, 
%\end{equation}
is a stationary solution to the SDE
\begin{equation}
\ud \Zh{t} = -\delta a \Zh{t}  \ud t  + \delta^H \ud \Wh{t}.
\end{equation}
It is a zero-mean, stationary Gaussian process with variance $\sigma_{ou}^2$ and covariance $\EE\left[\Zh{t}\Zh{t+s}\right] =\sigma^2_{ou}\MCC_Z(\delta s)$.
%\begin{equation}
%\EE\left[\Zh{t}\Zh{t+s}\right] =\sigma^2_{ou}\MCC_Z(\delta s) = \sigma_{ou}^2 \frac{2\sin(\pi H)}{\pi} \int_0^\infty \cos(\delta asx)\frac{x^{1-2H}}{1+x^2}\ud x.
%\end{equation}
The covariance function depends on $\delta s$ only, which indicates that $1/\delta$ is the natural scale of $\Zh{t}$ as desired. More properties and estimates regarding $\Zh{t}$ are stated in Lemma \ref{lem_moments}.

As $\delta$ goes to zero, by dominated convergence theorem and $\MCC_Z(0) = 1$, the covariance becomes
%\begin{equation}
%\lim_{\delta \to 0}\EE\left[\Zh{t}\Zh{t+s}\right] = \sigma_{ou}^2 \frac{2\sin(\pi H)}{\pi} \int_0^\infty \frac{x^{1-2H}}{1+x^2}\ud x = \sigma^2_{ou}\frac{2\sin(\pi H)}{\pi}\frac{\pi}{2}\csc(\pi H) = \sigma^2_{ou},
%\end{equation}
\begin{equation}
\lim_{\delta \to 0}\EE\left[\Zh{t}\Zh{t+s}\right] = \sigma_{ou}^2 \MCC_Z(0) = \sigma^2_{ou},
\end{equation}
and the process $\Zh{t}$ converges in distribution to $\left(\Zh{0}\right)_{t \in \RR} \stackrel{\MCD}{=} \left(\sigma_{ou}Z\right)_{t\in \RR}$ , where $Z$ is a standard normal random variable.

\subsection{First order Approximation to the Value Process}\label{sec_asymppower}
In this section, we study the problem discussed in Section~\ref{sec_power} with $Y_t = \Zh{t}$ and $W_t^Y = W_t^Z$. To be precise, the underlying asset $S_t$ is driven by the slowly varying fractional stochastic factor $\Zh{t}$ defined in \eqref{def_Zh}, 
\begin{equation*}
\ud S_t = \mu(\Zh{t}) S_t \ud t + \sigma(\Zh{t}) S_t \ud W_t.
\end{equation*}
Still, we denote by $X_t^\pi$ the wealth process, and it follows
\begin{equation*}
\ud X_t^\pi =\pi_t \mu(\Zh{t}) \ud t + \pi_t \sigma(\Zh{t})\ud W_t.
\end{equation*}
The value process is denoted by $V_t^\delta$ to indicate its dependence of $\delta$ introduced by the slowly varying process $\Zh{\cdot}$:
\begin{equation*}
V_t^\delta := \esssup_{\pi \in \MCA_t}\EE\left[U(X_T^\pi)\vert \MCF_t\right].
\end{equation*}

Note that, by definition, the process $\Zh{\cdot}$ is neither Markovian nor a semimartingale when $H \neq \half$, therefore the HJB equation is not available. However, it is adapted to $\mc{G}_t$. In order to use Proposition~\ref{thm_martdistort}, we need to check that $\Zh{t}$ satisfies Assumption~\ref{assump_power}(ii).

\begin{lem}
	The slowly varying fractional factor $\Zh{t}$ defined in \eqref{def_Zh} satisfies the Assumption~\ref{assump_power}(ii).
\end{lem}
\begin{proof}
It suffices to show that $\xi_t$ is bounded uniformly in $t\in [0,T]$ and $\delta$.
To obtain the process $(\xi_t)_{t\in[0,T]}$ in \eqref{def_xi}, we shall use Malliavin calculus. By the Clark-Ocone Formula (see \cite{DiOkPr:09}), we obtain 
$$M_t\xi_t = \widetilde \EE[\widetilde \MCD_t M_T \vert \MCG_t],$$
where $\widetilde \MCD_t$ denotes the Malliavian derivative with respect to the Brownian motion $$\widetilde W_t^Z = W_t^Z - \rho\left(\frac{1-\gamma}{\gamma}\right)\int_0^t \lambda(\Zh{s})\ud s.$$
The term $\widetilde \MCD_t M_T$ is computed as:
\begin{equation*}
\widetilde \MCD_t M_T = e^{\frac{1-\gamma}{2q\gamma}\int_0^T \lambda^2(\Zh{s})\ud s} \int_0^T \frac{1-\gamma}{q\gamma} \lambda(\Zh{s})\lambda'(\Zh{s})\widetilde \MCD_t \Zh{s} \ud s.
\end{equation*}
Since $M_t$, $\lambda$ and $\lambda'$ are bounded, it suffices to show $\int_0^T \abs{\widetilde \MCD_t \Zh{s}} \ud s$ to be uniformly bounded.

To this end, recall $\Zh{s}$ defined in \eqref{def_Zh}:
$$\Zh{s} = \int_{-\infty}^s \mc{K}^\delta(s-u)\ud W_u^Z = \int_{-\infty}^s \mc{K}^\delta(s-u)\ud \widetilde W_u^Z + \int_0^s \MCK^\delta(s-u)\rho\frac{1-\gamma}{\gamma}\lambda(\Zh{u})\ud u.$$
It is adapted to $\MCG_t$, thus $\widetilde \MCD_t \Zh{s} = 0$. for $t\geq s$. For $t < s$, we deduce
\begin{equation*}
\widetilde \MCD_t \Zh{s} = \MCK^\delta(s-t) + \int_t^s \MCK^\delta(s-u)\rho\frac{1-\gamma}{\gamma}\lambda'(\Zh{u})\widetilde \MCD_t \Zh{u} \ud u.
\end{equation*}
Therefore, by defining the positive increasing function $A^\delta(t) = \int_0^t \MCK^\delta(s)\ud s$, one has
\begin{align*}
\int_0^T \abs{\widetilde \MCD_t \Zh{s}} \ud s %&= \int_t^T \abs{\widetilde \MCD_t \Zh{s}} \ud s \\
& \leq \int_t^T \MCK^\delta(s-t) \ud s  + \abs{\rho\frac{1-\gamma}{\gamma}}\int_t^T \int_t^s \MCK^\delta(s-u)\abs{\lambda'(\Zh{u})} \abs{\widetilde \MCD_t \Zh{u}} \ud u \ud s \\
& \leq \int_0^{T-t} \MCK^\delta(s)\ud s +  \abs{\rho\frac{1-\gamma}{\gamma}} \infnorm{\lambda'}\int_t^T\int_u^T \MCK^\delta(s-u)  \abs{\widetilde \MCD_t \Zh{u}} \ud s \ud u \\
%& = A^\delta(T-t) + \rho\frac{1-\gamma}{\gamma} \infnorm{\lambda'}\int_t^T A^\delta(T-u) \abs{\widetilde \MCD_t \Zh{u}}  \ud u \\
%& \leq A^\delta(T-t) + \rho\frac{1-\gamma}{\gamma} \infnorm{\lambda'} A^\delta(T-t) \int_t^T  \abs{\widetilde \MCD_t \Zh{u}}  \ud u \\
& \leq A^\delta(T) + \abs{\rho\frac{1-\gamma}{\gamma}} \infnorm{\lambda'} A^\delta(T) \int_t^T  \abs{\widetilde \MCD_t \Zh{u}}  \ud u, 
\end{align*}
and for any $t \in [0,T]$,
\begin{equation*}
\int_0^T  \abs{\widetilde \MCD_t \Zh{s}}  \ud s \leq 
\frac{A^\delta(T)}{1- \abs{\rho\frac{1-\gamma}{\gamma}} \infnorm{\lambda'} A^\delta(T)}
\end{equation*}
provided $1- \abs{\rho\frac{1-\gamma}{\gamma}} \infnorm{\lambda'} A^\delta(T)$ is positive. This holds for sufficient small $\delta$ since $A^\delta(T)$ is of order $\delta^H$ (see Lemma \ref{lem_moments}\eqref{lem_psi}), which completes the proof.

\end{proof}

\begin{theo}\label{thm_Vtpowerexpansion}
Under Assumption \ref{assump_power}, for fixed $t \in [0,T)$, $X_t = x$ and the observed value $\Zh{0}$, $V_t^\delta$ takes the form
\begin{equation}\label{eq_Vtpower}
V_t ^\delta = Q^\delta_t(X_t, \Zh{0}) + \MCO(\delta^{2H}),
\end{equation}
where 
\begin{equation}\label{def_Qdelta}
Q_t^\delta(x,z) = \frac{x^{1-\gamma}}{1-\gamma}e^{\frac{1-\gamma}{2\gamma}\lambda^2(z)(T-t)}\left[1  + \frac{1-\gamma}{\gamma}\lambda(z)\lambda'(z)\left( \phi_t^\delta +  \delta^H\rho \lambda(z)\left(\frac{1-\gamma}{\gamma}\right) \frac{(T-t)^{H+\frac{3}{2}}}{\Gamma(H + \frac{5}{2})}\right)\right].
\end{equation}
Here  $\phi_t^\delta$ is defined by
\begin{equation}\label{def_phi}
\phi_t^\delta = \EE\left[\left.\int_t^T \left(\Zh{s} - \Zh{0}\right) \ud s \right\vert \MCF_t\right] = \EE\left[\left.\int_t^T \left(\Zh{s} - \Zh{0}\right) \ud s \right\vert \MCG_t\right],
\end{equation}
and $\phi_t^\delta$ is of order $\delta^H$ as proved in Lemma~\ref{lem_moments} in the sense that its variance is of order $\delta^{2H}$. Note that $O(\delta^{2H})$ denotes a $\MCF_t$-adapted random variable and it is of order $\delta^{2H}$ in $L^2$.
\end{theo}

\begin{proof}

A straightforward application of Proposition~\ref{thm_martdistort} with $Y_t = \Zh{t}$ and $W_t^Y = W_t^Z$ gives the following representation of the value process $V_t^\delta$ 
\begin{equation*}
	V_t^\delta=\frac{X_t^{1-\gamma}}{1-\gamma} \left[\widetilde{\EE}\left(\left.e^{\frac{1-\gamma}{2q\gamma}\int_t^T\lambda^2(\Zh{s})\ud s}\right\vert \mc{G}_t\right)\right]^q.
\end{equation*}
We start by expanding $\Psi_t^\delta:= \widetilde \EE\left[\left.e^{\frac{1-\gamma}{2q\gamma}\int_t^T \lambda^2(\Zh{s})\ud s}\right\vert \MCG_t\right]$, and then apply Taylor formula to the function $x^q$. 

The formula for the conditional expectation under an absolute continuous change of measure, together with the value of $a_t$ given by \eqref{def:at}  and Taylor expansion in $z$ at the point $\Zh{0}$ yields,
\begin{align*}
\Psi_t^\delta& = \EE\left[\left.e^{\frac{1-\gamma}{2q\gamma}\int_t^T \lambda^2(\Zh{s})\ud s}e^{-\int_t^T a_s\ud W_s^Z - \half \int_t^T a_s^2 \ud s}\right\vert \MCG_t\right] \\
& = \EE\left[\left.e^{\frac{1-\gamma}{2q\gamma}\int_t^T \lambda^2(\Zh{s})\ud s}e^{\int_t^T \rho \left(\frac{1-\gamma}{\gamma}\right)\lambda(\Zh{s})\ud W_s^Z - \half \int_t^T \rho^2\left(\frac{1-\gamma}{\gamma}\right)^2\lambda^2(\Zh{s}) \ud s}\right\vert \MCG_t\right] \\
& = e^{\frac{1-\gamma}{2q\gamma}\lambda^2(\Zh{0})(T-t)} \EE\left[\left.e^{\int_t^T \rho\left(\frac{1-\gamma}{\gamma}\right)\lambda(\Zh{0})\ud W_s^Z - \half \int_t^T \rho^2 \left(\frac{1-\gamma}{\gamma}\right)^2\lambda^2(\Zh{0})\ud s + A_{[t,T]} + B_{[t,T]}}\right\vert\MCG_t \right],
\end{align*}
where $A_{[t,T]}$ and $B_{[t,T]}$ are given by
\begin{align*}
A_{[t,T]} =& \frac{1-\gamma}{q\gamma}\lambda(\Zh{0})\lambda'(\Zh{0}) \int_t^T \left(\Zh{s} - \Zh{0}\right)\ud s + \rho\left(\frac{1-\gamma}{\gamma}\right)\lambda'(\Zh{0})\int_t^T\left(\Zh{s}-\Zh{0}\right)\ud W_s^Z
\\
&- \rho^2\left(\frac{1-\gamma}{\gamma}\right)^2\lambda(\Zh{0})\lambda'(\Zh{0})\int_t^T \left(\Zh{s}-\Zh{0}\right) \ud s,\\
B_{[t,T]} = & \frac{1-\gamma}{q\gamma} \int_t^T \left(\lambda\lambda'' + \lambda'^2\right)(\chi_s)\left(\Zh{s} - \Zh{0}\right)^2\ud s + \rho\left(\frac{1-\gamma}{\gamma}\right)\int_t^T\lambda''(\eta_s)\left(\Zh{s}-\Zh{0}\right)^2\ud W_s^Z
\\
&- \rho^2\left(\frac{1-\gamma}{\gamma}\right)^2\int_t^T \left(\lambda\lambda'' + \lambda'^2\right)(\chi_s)\left(\Zh{s}-\Zh{0}\right)^2 \ud s,
\end{align*}
with $\chi_s$ and $\eta_s$ being the Lagrange remainders: $\chi_s$, $\eta_s \in [\Zh{0}\vee\Zh{s}, \Zh{0}\wedge \Zh{s}]$. 

Since $\lambda(\cdot)$ is bounded, one can expand $e^{A_{[t,T]} + B_{[t,T]}}$ and deduce
\begin{align*}
\Psi_t^\delta =& e^{\frac{1-\gamma}{2q\gamma}\lambda^2(\Zh{0})(T-t)} \EE\left[\left.e^{\int_t^T \rho\left(\frac{1-\gamma}{\gamma}\right)\lambda(\Zh{0})\ud W_s^Z - \half \int_t^T \rho^2 \left(\frac{1-\gamma}{\gamma}\right)^2\lambda^2(\Zh{0})\ud s}\left(1 + A_{[t,T]} + R_{[t,T]}\right)\right\vert\MCG_t \right]\\
= &e^{\frac{1-\gamma}{2q\gamma}\lambda^2(\Zh{0})(T-t)} \EE\left[\left.e^{\int_t^T \rho\left(\frac{1-\gamma}{\gamma}\right)\lambda(\Zh{0})\ud W_s^Z - \half \int_t^T \rho^2\left(\frac{1-\gamma}{\gamma}\right)^2\lambda^2(\Zh{0}) \ud s}\left(1 + A_{[t,T]} \right)\right\vert\MCG_t \right] + \MCO(\delta^{2H}),
\end{align*}
where $R_{[t,T]}$ is given by
\begin{equation}\label{def_RtT}
R_{[t,T]} = e^{A_{[t,T]} + B_{[t,T]}} - 1 - A_{[t,T]}.
\end{equation}
In Lemma \ref{lem_RtT} it is proved that $R_{[t,T]} \sim \MCO(\delta^{2H})$. As mentioned before, we denote by $O(\delta^{2H})$ a random variable of order $\delta^{2H}$ in $L^2$ sense.

We introduce a new probability measure $\widehat \PP$, such that under $\widehat \PP$, $ \widehat W_t^Z =  W_t^Z - \rho\left(\frac{1-\gamma}{\gamma}\right)\lambda(\Zh{0})t$ is a standard Brownian motion. Then $\Psi_t^\delta$ can be rewritten as
\begin{align*}
\Psi_t^\delta = &e^{\frac{1-\gamma}{2q\gamma}\lambda^2(\Zh{0})(T-t)} \widehat \EE\left[\left.\left(1 + A_{[t,T]} \right)\right\vert\MCG_t \right] + \MCO(\delta^{2H})\\
= & e^{\frac{1-\gamma}{2q\gamma}\lambda^2(\Zh{0})(T-t)} \widehat\EE\left[\left.
%e^{\int_t^T \frac{1-\gamma}{\gamma}\lambda(\Zh{0})\rho\ud W_s - \half \int_t^T \left(\frac{1-\gamma}{\gamma}\right)^2\lambda^2(\Zh{0})\rho^2 \ud s}
1+ \frac{(1-\gamma)}{q\gamma}\lambda(\Zh{0})\lambda'(\Zh{0})\int_t^T \left(\Zh{s} - \Zh{0}\right)\ud s\right\vert \MCG_t \right]\\
& +e^{\frac{1-\gamma}{2q\gamma}\lambda^2(\Zh{0})(T-t)} \widehat \EE\left[
\left. \rho\left(\frac{1-\gamma}{\gamma}\right)\lambda'(\Zh{0})\int_t^T \left(\Zh{s}-\Zh{0}\right)\ud W_s^Z  \right\vert \MCG_t\right]\\
& -e^{\frac{1-\gamma}{2q\gamma}\lambda^2(\Zh{0})(T-t)} \widehat \EE\left[
\left. \rho^2\left(\frac{1-\gamma}{\gamma}\right)^2\lambda(\Zh{0})\lambda'(\Zh{0})\int_t^T \left(\Zh{s}-\Zh{0}\right)\ud s \right\vert \MCG_t\right] + O(\delta^{2H}),
\end{align*}
and the second term cancels with the third one, since
\begin{align*}
\widehat \EE & \left[
\left. \rho\left(\frac{1-\gamma}{\gamma}\right)\lambda'(\Zh{0})\int_t^T \left(\Zh{s}-\Zh{0}\right)\ud W_s^Z  \right\vert \MCG_t\right] \\
=  &\widehat \EE\left[
\left. \rho\left(\frac{1-\gamma}{\gamma}\right)\lambda'(\Zh{0})\int_t^T \left(\Zh{s}-\Zh{0}\right)\ud \widehat W_s^Z  \right\vert \MCG_t\right] \\
& +  \widehat \EE\left[
\left. \rho\left(\frac{1-\gamma}{\gamma}\right)\lambda'(\Zh{0})\int_t^T \left(\Zh{s}-\Zh{0}\right) \rho\left(\frac{1-\gamma}{\gamma}\right)\lambda(\Zh{0})\ud s  \right\vert \MCG_t\right]\\
 = &  \widehat\EE\left[
\left. \rho^2\left(\frac{1-\gamma}{\gamma}\right)^2\lambda(\Zh{0})\lambda'(\Zh{0})\int_t^T \left(\Zh{s}-\Zh{0}\right)\ud s \right\vert \MCG_t\right].
\end{align*}
Thus, the term $\Psi^\delta_t$ is simplified to
\begin{equation}\label{eq_I}
\Psi^\delta_t = e^{\frac{1-\gamma}{2q\gamma}\lambda^2(\Zh{0})(T-t)}\left(1 + \frac{(1-\gamma)}{q\gamma}\lambda(\Zh{0})\lambda'(\Zh{0}) \Phi^\delta_t\right) + \MCO(\delta^{2H}),
\end{equation}
 with
\begin{align*}
\Phi^\delta_t = \widehat\EE\left[\left.
\int_t^T\left(\Zh{s} - \Zh{0}\right)\ud s\right\vert \MCG_t \right] =  \widehat\EE\left[\left.
\int_t^T\Zh{s}\ud s\right\vert \MCG_t \right] - \Zh{0}(T-t).
\end{align*}
To further simplify $\Phi^\delta_t$, we use the moving average representation \eqref{def_Zh} for $\Zh{s}$ 
%\begin{equation*}
%\Zh{s} = \int_{-\infty}^s \MCK^{\delta}(s-u) \ud W_u,
%\end{equation*}
and deduce
\begin{align}
\Phi^\delta_t & =  \widehat\EE\left[\left.
\int_t^T\Zh{s}\ud s\right\vert \MCG_t \right] - \Zh{0}(T-t) = \widehat\EE\left[\left.
\int_t^T\int_{-\infty}^{s}\MCK^\delta(s-u) \ud W_u^Z\ud s\right\vert \MCG_t \right] - \Zh{0}(T-t)\nonumber\\
& = \widehat\EE\left[\left.
\int_{-\infty}^t\int_{t}^{T}\MCK^\delta(s-u) \ud s \ud  W_u^Z\right\vert \MCG_t \right] + \widehat\EE\left[\left.
\int_{t}^T\int_{u}^{T}\MCK^\delta(s-u) \ud s \ud  W_u^Z\right\vert \MCG_t \right]  - \Zh{0}(T-t)\nonumber \\
& =\int_{-\infty}^t\int_{t}^{T}\MCK^\delta(s-u) \ud s \ud  W_u^Z - \Zh{0}(T-t) + \widehat\EE\left[\left.
\int_{t}^T\int_{u}^{T}\MCK^\delta(s-u) \ud s \ud  W_u^Z\right\vert \MCG_t \right]   \nonumber\\
& =\phi_t^\delta + \widehat\EE\left[\left.
\int_{t}^T\int_{u}^{T}\MCK^\delta(s-u) \ud s \ud  \widehat W_u^Z\right\vert \MCG_t \right] + \rho\left(\frac{1-\gamma}{\gamma}\right)\lambda(\Zh{0})\int_t^T \int_{u}^T \MCK^\delta(s-u)\ud s \ud u \nonumber\\
& = \phi_t^\delta +  \rho\left(\frac{1-\gamma}{\gamma}\right)\lambda(\Zh{0}) \frac{\delta^H(T-t)^{H+3/2}}{\Gamma(H+\frac{5}{2})} + O(\delta^{H+1}). \label{eq_II}
%= \phi_t^\delta + \delta^H \frac{1-\gamma}{\gamma}\lambda(\Zh{0})\rho D_{t,T} + O(\delta^{H+1}),
\end{align}
In the derivation, we have changed the order of $\ud s$ and $\ud W_u^Z$ and use the relation $ \widehat W_t^Z =  W_t^Z - \frac{1-\gamma}{\gamma}\lambda(\Zh{0})\rho t$. The change of order is justified by the stochastic Fubini theorem, for which a sufficient condition is 
\begin{equation*}
	\int_t^T \left(\int_{-\infty}^s \MCK^\delta(s-u)^2 \ud u \right)^{1/2} \ud s < \infty, \quad \MCK^\delta(t) = \sqrt{\delta}\MCK(\delta t).
\end{equation*}	
This follows by $\MCK \in L^2(0, \infty)$. Now combining \eqref{eq_I} and \eqref{eq_II}, we obtain
\begin{align*}
V_t^\delta =& \frac{X_t^{1-\gamma}}{1-\gamma} \left(\Psi^\delta_t\right)^q \\
= &\frac{X_t^{1-\gamma}}{1-\gamma}e^{\frac{1-\gamma}{2\gamma}\lambda^2(\Zh{0})(T-t)}\left\{1 + \frac{1-\gamma}{\gamma}\lambda(\Zh{0})\lambda'(\Zh{0})\Phi^\delta_t\right\} + O(\delta^{2H})\\
=& \frac{X_t^{1-\gamma}}{1-\gamma}e^{\frac{1-\gamma}{2\gamma}\lambda^2(\Zh{0})(T-t)}\left\{1 + \frac{1-\gamma}{\gamma}\lambda(\Zh{0})\lambda'(\Zh{0})\left(\phi_t^\delta +  
 \delta^H\rho \lambda(\Zh{0})\left(\frac{1-\gamma}{\gamma}\right) \frac{(T-t)^{H+\frac{3}{2}}}{\Gamma(H + \frac{5}{2})}\right)\right\}  \\
 &+ O(\delta^{2H}).
%\\& = \vz(t, X_t, \Zh{0}) + \lambda(\Zh{0})\lambda'(\Zh{0})D_1\vz(t,X_t,\Zh{0})\phi_t^\delta + \delta^H\rho\lambda^2(\Zh{0})\lambda'(\Zh{0})\vo(t,X_t,\Zh{0}) + o(\delta^H).
\end{align*}

Observe that there are two corrections to the leading term: a random component $\phi_t^\delta$, and a deterministic function of $(t, X_t, \Zh{0})$, both being of order $\delta^H$.
\end{proof}
\begin{rem}[Discussion of the assumptions on $\lambda(\cdot)$]
	
In order to expand $\Psi_t^\delta$, we need a uniform bound (in $\delta$) of $\EE\left[e^{\frac{1-\gamma}{2q\gamma}\int_t^T \lambda^2(\Zh{s}) \ud s }\right]$.	
Notice that if $\gamma>1$, this is automatically satisfied, since the exponential function is bounded by 1. For $0<\gamma<1$, it is also satisfied under the assumption $\lambda(\cdot)$ bounded as stated in Assumption~\ref{assump_power}(i). Moreover, the assumption can be relaxed to have uniform bounds for exponential moments of the function $\lambda^2(\cdot)$.
\end{rem}

\subsection{Optimal Strategy}\label{sec_asymppi}
We now turn to the expansion to the optimal portfolio given in \eqref{def_pioptimal}
\begin{equation*}
\pi^\ast_t = \left[\frac{\lambda(\Zh{t})}{\gamma \sigma(\Zh{t})} + \frac{\rho q \xi_t}{\gamma \sigma(\Zh{t})}\right] X_t,  
\end{equation*}
where the process $\xi_t$ given by the representation theorem \eqref{def_xi} is usually not known explicitly. In this section, we approximate $\xi_t$ using the results derived in Theorem~\ref{thm_Vtpowerexpansion}, and we obtain the following asymptotic result for $\pi^\ast_t$.
\begin{theo}\label{thm_piexpansion}
Under Assumption~\ref{assump_power}, the optimal strategy $\pi_t^\ast$ is approximated by
\begin{align}\label{eq_piapprox}
\pi^\ast_t &=  \left[\frac{\lambda(\Zh{t})}{\gamma \sigma(\Zh{t})} +\delta^H \frac{\rho(1-\gamma)}{\gamma^2 \sigma(\Zh{t})}\frac{(T-t)^{H+1/2}}{\Gamma(H+\frac{3}{2})}\lambda(\Zh{0})\lambda'(\Zh{0}) \right] X_t + \MCO(\delta^{2H}) \\
 &:= \pi_t^{(0)} + \delta^H \pi_t^{(1)} + \MCO(\delta^{2H}). \nonumber
\end{align}
\end{theo}

Before proving this theorem, we give some important remarks. 
\begin{rem}\label{rem_piapprox}
	\quad 
	
\begin{enumerate}[(i)]
\item 
For the case $H=\half$, $\Zh{t}$ becomes the Markovian OU process, and \eqref{eq_piapprox} coincides with the approximation of feedback form derived in \cite[Section~3.2.2 and 6.3.2]{FoSiZa:13}.

\item 
In the approximation \eqref{eq_piapprox} to $\pi^\ast_t$, the leading order strategy $\pi_t^{(0)}$  follows the process $\Zh{t}$, the first order correction $\pi_t^{(1)}$ is partially frozen at $\Zh{0}$, and the random correction $\phi_t^\delta$ appearing in $V_t$ disappears here. This makes the approximated strategy $\pi_t^{(0)} + \delta^H \pi_t^{(1)}$ easier to implement.

Moreover, under additional  smoothness assumption on $\sigma(\cdot)$, typically $\sigma(\cdot)$ is $C^1$ and $(1/\sigma(\cdot))'$ is bounded, 
%then 
%if the function $\sigma(\cdot)$ is $C^1$ and satisfies \todo{To be checked, the derivative of $1/\sigma$ is bounded would work}
%\begin{equation*}
%\frac{\sigma'(\chi_t)}{\sigma^2(\chi_t)} \in L^2, \quad \chi_t \text{ is the Lagrange remainder from } \frac{1}{\sigma(\Zh{t})} = \frac{1}{\sigma(\Zh{0})} + \frac{\sigma'(\chi_t)}{\sigma^2(\chi_t)}(\Zh{t} - \Zh{0}),
%\end{equation*}
then the correction term $ \pi_t^{(1)}$ can be fully frozen at $\Zh{0}$ without changing the order of accuracy, namely,
\begin{equation*}
 \pi_t^{(1)} = \frac{\rho(1-\gamma)}{\gamma^2 \sigma(\Zh{0})}\frac{(T-t)^{H+1/2}}{\Gamma(H+\frac{3}{2})}\lambda(\Zh{0})\lambda'(\Zh{0})X_t + \MCO(\delta^{H}).
\end{equation*}
	
\item 

Denote by $X_t^\pz$ the wealth process following the zeroth order strategy $\pi_t^{(0)} = \frac{\lambda(\Zh{t})}{\gamma\sigma(\Zh{t})}X_t$ 
\begin{equation}\label{rem_Xtpz}
\ud X_t^\pz = \mu(\Zh{t})\pi_t^{(0)} \ud t + \sigma(\Zh{t})\pi_t^{(0)} \ud W_t,
\end{equation}
and $\Vzl_\cdot$ the corresponding value process 
\begin{equation*}
 \Vzl_t:= \EE\left[\left.U\left( X_T^\pz \right)\right\vert \MCF_t\right].
 %= U(X_t^\pz)\EE[e^{\frac{(2\gamma-1)(1-\gamma)}{2\gamma^2}\int_t^T \lambda^2(\Zh{s}) \ud s + \frac{1-\gamma}{\gamma} \int_t^T \lambda(\Zh{s})\ud W_s}\vert \MCF_t].
% \quad U(x) = \frac{x^{1-\gamma}}{1-\gamma}.
\end{equation*}
In Section~\ref{sec_decompz} Proposition~\ref{prop_Vz}, we derive the expansion to $\Vzl_t$ for general utility function. When applied to the case of power utility \eqref{def_power}, one can deduce that $\Vzl - Q^\delta_t$ is of order $\delta^{2H}$ with $Q_t^\delta$ given in \eqref{def_Qdelta}. Therefore, by Theorem~\ref{thm_Vtpowerexpansion}, $\Vzl_t- \Vz_t$ is of order $\delta^{2H}$, and we conclude that  $\pi_t^{(0)} = \frac{\lambda(\Zh{t})}{\gamma\sigma(\Zh{t})}X_t$ generates the approximated value process given by \eqref{eq_Vtpower}, and is asymptotically optimal within all admissible strategy $\MCA_t$ up to order $\delta^H$.
%Introduce Xt followsing pz, Vpz. 
%Using the eps-martinge, this is done wiht U in section 4. 
%When applied to the case power utility, we  will deduce that Vpz - Qgamma is of order delta2H. Therefore, by theorem 3.1, Vpz - V is of order delta2H, and conclude blablabal.
 
\end{enumerate}

\end{rem}

\begin{proof}[Proof of Theorem \ref{thm_piexpansion}]
It suffices to derive the expansion of $\xi_t$ determined by \eqref{def_xi}. 
In the previous section, we have obtained a rigorous expansion for $\Psi^\delta_t:= \widetilde \EE\left[\left.e^{\frac{1-\gamma}{2q\gamma}\int_t^T \lambda^2(\Zh{s})\ud s}\right\vert \MCG_t\right]$; see \eqref{eq_I} and \eqref{eq_II}. Rewrite $M_t$ defined in \eqref{def_Mtmartdistort} using $\Psi^\delta_t$ as 
\begin{equation*}
M_t = e^{ I_t} \Psi^\delta_t,
\end{equation*}
where $I_t = \frac{1-\gamma}{2q\gamma}\int_0^t\lambda^2(\Zh{s}) \ud s$.
Applying It\^o's formula to $M_t$ yields,
\begin{align*}
\ud M_t =& e^{I_t} \Psi^\delta_t \ud I_t+ e^{I_t} \ud \Psi_t^\delta \\
 = &  \frac{1-\gamma}{2q\gamma}\lambda^2(\Zh{t})M_t \ud t + e^{ I_t}\left(-\frac{1-\gamma}{2q\gamma}\lambda^2(\Zh{0}) \Psi_t^\delta \ud t +  e^{\frac{1-\gamma}{2q\gamma}\lambda^2(\Zh{0})(T-t)} \frac{1-\gamma}{q\gamma}\lambda(\Zh{0})\lambda'(\Zh{0})\ud \Phi_t^\delta\right) \\
&+ \MCO(\delta^{2H})\\
 = & \delta^H e^{I_t} e^{\frac{1-\gamma}{2q\gamma}\lambda^2(\Zh{0})(T-t)} \frac{1-\gamma}{q\gamma}\lambda(\Zh{0})\lambda'(\Zh{0}) \theta_{t,T}\ud \widetilde W_t^Z + \MCO(\delta^{2H}).
\end{align*}
Here in the derivation, we have successively used the relation \eqref{eq_I} and \eqref{eq_II}, $\ud \psi_t^\delta = \ud \phi_t^\delta + (\Zh{t}-\Zh{0})\ud t $, where $\psi_t^\delta$ is given by
\begin{equation}\label{def_psi}
\psi_t^\delta = \EE\left[\left.\int_0^T \Zh{s}-\Zh{0}\ud s \right\vert \MCF_t\right]
\end{equation}
and $\ud \psi_t^\delta = \delta^H\theta_{t,T}\ud W_t^Z + \delta^{H+1}\widetilde \theta_{t,T}\ud W_t^Z$ with $\theta_{t,T}$ and $\widetilde \theta_{t,T}$ specified in Lemma~\ref{lem_moments}.

Noticing that from \eqref{eq_I}, one can deduce
\begin{equation*}
\Psi_t^\delta  = e^{\frac{1-\gamma}{2q\gamma}\lambda^2(\Zh{0})(T-t)} + \MCO(\delta^H),
\end{equation*}
then $\ud M_t$ becomes 
\begin{align*}
\ud M_t &=  \delta^H e^{ I_t}\Psi_t^\delta \frac{1-\gamma}{q\gamma}\lambda(\Zh{0})\lambda'(\Zh{0}) \theta_{t,T}\ud \widetilde W_t^Z + \MCO(\delta^{2H})\\
& = \left[ \delta^H \frac{1-\gamma}{q\gamma}\lambda(\Zh{0})\lambda'(\Zh{0}) \frac{(T-t)^{H+\half}}{\Gamma(H + \frac{3}{2})} \right]M_t \ud \widetilde W_t^Z + \MCO(\delta^{2H}),
\end{align*}
and the approximation of $\xi_t$ is given by
\begin{equation*}
\xi_t = \delta^H \frac{1-\gamma}{q\gamma}\lambda(\Zh{0})\lambda'(\Zh{0}) \frac{(T-t)^{H+\half}}{\Gamma(H + \frac{3}{2})} + \MCO(\delta^{2H}).
\end{equation*}
Plugging the above expression into \eqref{def_pioptimal} yields the desired result \eqref{eq_piapprox}.
\end{proof}

\subsection{Numerical illustration}\label{sec_num}

Next, we illustrate numerically the asymptotic optimality property of $\pi_t^{(0)}$ mentioned in Remark \ref{rem_piapprox}(iii). That is, we compute $\Vz_t$ and $\Vzl_t$  at time $t=0$ using Monte Carlo simulations, and compare their differences. Using equation \eqref{def:Vcorrelated} and changing the measure from $\widetilde \PP$ to $\PP$ give
\begin{align*}
\Vz_0 =  \frac{X_0^{1-\gamma}}{1-\gamma}\left[\EE\left(e^{\left(\frac{1-\gamma}{2\gamma}\right)\int_0^T \lambda^2(\Zh{s})\ud s + \rho\left(\frac{1-\gamma}{\gamma}\right)\int_0^T \lambda(\Zh{s})\ud W_s^Z} \Big\vert \MCG_0\right)\right]^q.
\end{align*} 
Solving the SDE \eqref{rem_Xtpz} for $X_t^\pz$ and plugging the solution into the definition of $\Vzl_t$ yield
\begin{align*}
\Vzl_0 = \frac{X_0^{1-\gamma}}{1-\gamma}\EE\left(e^{\left(\frac{-2\gamma^2 + 3\gamma - 1}{2\gamma^2}\right)\int_0^T \lambda^2(\Zh{s})\ud s + \left(\frac{1-\gamma}{\gamma}\right)\int_0^T \lambda(\Zh{s})\ud W_s} \Big\vert \MCF_0\right).
\end{align*}

The model parameters are chosen as:
\begin{equation*}
T = 1, \quad H = 0.1, \quad a = 1, \quad  \gamma = 0.4, \quad \rho = -0.5, \quad \mu(y) = \frac{0.1 \times \lambda(y)}{0.1 + \lambda(y)}, \quad  \lambda^2(y) = \half \int_{-\infty}^{y/\sigma_{ou}}p(z/2)\ud z,
\end{equation*}
where we recall that $p(z)$ is the $\mc{N}(0,1)$-density.
Note that the choice of $\lambda(y)$ above satisfies the model Assumption \ref{assump_power}.

Due to the natural non-Markovian structure, we first generate a ``historical'' path $W_t^Z$ between $-M$ and $0$, and then evaluate each conditional expectation by the average of 500,000 paths. The slow factor $(\Zh{t})_{t \in [0,T]}$ is generated using Euler scheme with mesh size $\Delta t = 10^{-3}$, and $M = (T/\Delta t)^{0.5}\Delta t$ (due to short-range dependence). 

The numerical results presented in Table \ref{table1} are only for a purpose of illustration as we computed the values for only a few ``omegas" denoted by $\#1, \#2,$, $\#3$, $\#4$ and $\#5$.

\begin{table}[H]
	\centering
	\caption{The value processes $\Vz_0$ \emph{vs.} $\Vzl_0$  for the power utility case.}\label{table1}
	\begin{tabular}{|c|c|c|c|c|c|c|}\hline
		&& \#1 & \#2 & \#3  & \#4 & \#5 \\ \hline\hline 
		&$\Vz_0$ & 1.4645 & 1.4067 & 1.4253 & 1.4212 & 1.4082\\
		$\delta = 1$ & $\Vz_0 - \Vzl_0$& 0.0021 & 0.0021 & 0.0020 & 0.0020 & 0.0020 \\ \hline
		
		&$\Vz_0$& 1.4739 & 1.3995 & 1.4237& 1.4188 & 1.4019 \\
		$\delta = 0.5$ & $\Vz_0 -\Vzl_0$& 0.0022 & 0.0022 & 0.0022 & 0.0022 & 0.0023\\ \hline
		
		&$\Vz_0$ & 1.4814 & 1.3972 & 1.4248 & 1.4195 & 1.4002\\
		$\delta = 0.1$ & $\Vz_0 - \Vzl_0$ & 0.0020 & 0.0022 & 0.0022 & 0.0022 & 0.0022\\ \hline
		
		&$\Vz_0$ & 1.4811& 1.3990 & 1.4260 & 1.4208 & 1.4020\\
		$\delta = 0.05$ & $\Vz_0 -\Vzl_0$ & 0.0019 & 0.0020 & 0.0021 & 0.0020 & 0.0021 \\ \hline
		
		&$\Vz_0$ & 1.4783 & 1.4050 & 1.4291 & 1.4245 & 1.4076\\ 
		$\delta = 0.01$ & $\Vz_0 - \Vzl_0$ & 0.0016 & 0.0018 & 0.0018 & 0.0017 & 0.0018\\ \hline

	\end{tabular}
\end{table}

As expected, the strategy $\pi_t^{(0)}$ performs well for $\delta$ small, as the relative difference $(\Vz_0 - \Vzl_0)/\Vz_0$ is about $0.1\%$. What is more surprising is that it also performs well even for not so small values of $\delta$.

\section{General Utilities and Fractional Stochastic Environment}\label{sec_optimality}
In this section, we study the nonlinear portfolio optimization through asymptotics with general utility $U(x)$, and when the drift $\mu$ and volatility $\sigma$ of the underlying asset $S_t$ are driven by a slowly varying fractional stochastic factor $\Zh{t}$ defined in \eqref{def_Zh}. This is motivated by two recent works: in \cite{FoHu:16}, we developed asymptotic results for the value function following a given strategy in the slowly varying Markovian environment, and proved the optimality of such a strategy up to $o(\delta^H)$; on the other hand, asymptotics of linear pricing problem has been done and implied volatility is provided in \cite{GaSo:15} when the volatility is driven by $\Zh{t}$.

%For power utility \eqref{def_power} under slowly varying fractional factor \eqref{def_Zh}, we have rigorously derived the leading order strategy $\pi_t^{(0)} = \frac{\lambda(\Zh{t})}{\gamma\sigma(\Zh{t})} X_t$ to the optimal one $\pi_t^\ast$. 

Using the notation $M(t,x;\lambda)$ for the classical Merton value with constant Sharpe-ratio $\lambda$, we denote by $\vz$ the value function at frozen Sharpe-ratio $\lambda(z)$,
\begin{equation}\label{def_vz}
	\vz(t,x,z) = M(t,x,\lambda(z)).
\end{equation}
Then we define the strategy $\pz$ by 
	\begin{equation}\label{def_pz}
	\pz(t,x,z)  = -\frac{\lambda(z)}{\sigma(z)}\frac{\vz_x(t, x, z)}{\vz_{xx}(t, x, z)},
	\end{equation}
 and the associate value process $\Vzl$ is 
\begin{equation}\label{def_Vz}
\Vzl := \EE\left[\left.U\left( X_T^\pz \right)\right\vert \MCF_t\right],
\end{equation}
where $X_t^\pz$ is the wealth process following strategy $\pz$:
\begin{equation}\label{def_Xpz}
\ud X_t^\pz = \mu(\Zh{t})\pz(t,X_t^\pz, \Zh{t})\ud t + \sigma(\Zh{t})\pz(t,X_t^\pz, \Zh{t})\ud W_t.
\end{equation}

We first derive the expansion for $\Vzl$, and then we show that $\pz$ is optimal up to order $\delta^H$ among the strategies of the form 
\begin{equation}\label{def_Atilde}
\widetilde \MCA_t^\delta[\pzt, \pot, \alpha] := \left\{\pi = \pzt + \delta^\alpha \pot: \pi \in \MCA_t^\delta, \alpha >0, 0< \delta \leq 1 \right\},
\end{equation}
where $\MCA_t^\delta$ is the class of admissible controls \eqref{def_MCA} under the slowly varying fractional stochastic environment $\Zh{t}$. Motivated by the feedback form of $\pz$ in the power utility case \eqref{eq_piapprox} and definition of $\pz$ for general utility \eqref{def_pz}, we here restrict $\pzt$ and $\pot$ to be feedback controls. That is, $\pzt$, $\pot$ are functions of $(t,X_t,\Zh{t})$.
As a byproduct, by applying the expansion results for $\Vzl$ to power utility, $\pz$ obtained in Theorem~\ref{thm_piexpansion} is optimal up to order $\delta^{2H}$ within the full class of strategies $\MCA_t^\delta$.

%In this section, we first intend to study the performance of this strategy by giving the approximation to the value process $\Vzl$ associated to $\pz$:

%Actually this will be done in Proposition~\ref{prop_Vz} under general utility functions, and the form of $\pz$ used will be \eqref{def_pz}. Then, as stated in Remark~\ref{rem_piapprox}, with power utility \eqref{def_power}, this gives the asymptotic optimality of $\pz$ within all admissible strategies up to order $\delta^H$. Finally, the asymptotic optimality result is generalized to the case of general utility, but within a smaller class of strategies $\widetilde \MCA_t^\delta[\pzt, \pot, \alpha]$:

In the next subsection, we first review the classical Merton problem when $\mu$ and $\sigma$ are constants in \eqref{def_StunderY}, which plays a crucial role in deriving the expansion \eqref{eq_Vtexpansion} to $\Vzl$. Then we define some notations for later use.

\subsection{Merton Problem with Constant Coefficients}\label{sec_merton}
This problem has been extensively studied, for example, in \cite{KaSh:98}. Here we summarize the results about the classical Merton value function $M(t,x;\lambda)$. 

%Following the notations in \cite{FoSiZa:13}, we denote by $M(t,x;\lambda)$ the Merton value function and summarize the classical result of this value function.
 
 Assume that the utility function $U(x)$ is $C^2(0,\infty)$, strictly increasing, strictly concave, and satisfies the Inada and Asymptotic Elasticity conditions:
 \begin{equation*}
 U'(0+) = \infty, \quad U'(\infty) = 0, \quad \text{AE}[U] := \lim_{x\rightarrow \infty} x\frac{U'(x)}{U(x)} <1,
 \end{equation*}
 then, the Merton value function $M(t,x;\lambda)$ is strictly increasing, strictly concave in the wealth variable $x$, and decreasing in the time variable $t$, which is $C^{1,2}([0,T]\times \RR^+)$ and solves the HJB equation
 \begin{equation}\label{eq_value}
 M_t+\sup_{\pi}\left\{\frac{1}{2}\sigma^2\pi^2M_{xx}+\mu\pi M_x\right\}=
 M_t -\frac{1}{2}\lambda^2\frac{M_x^2}{M_{xx}} = 0, \quad M(T,x;\lambda) = U(x),
 \end{equation}
 where $\lambda = \mu/\sigma$ is the constant Sharpe ratio. It is  $C^1$ with respect to $\lambda$, and the optimal strategy is
 \begin{equation}\label{eq_pistar}
 \pi^\ast(t,x;\lambda)=-\frac{\lambda}{\sigma}\frac{M_x(t,x;\lambda)}{M_{xx}(t,x;\lambda)}.
 \end{equation}

 Given the Merton value function $M(t,x;\lambda)$, one can define the risk-tolerance function by
 \begin{equation}\label{def_risktolerance}
 R(t,x;\lambda) = -\frac{M_x(t,x;\lambda)}{ M_{xx}(t,x;\lambda)}.
 \end{equation}
 It is clear that $R(t,x;\lambda)$ is continuous and strictly positive due to the regularity, concavity and monotonicity of $M (t,x;\lambda)$. It is also smooth as a function of $\lambda$, see Remark \ref{rem_U} below.
  For further properties, we refer to \cite{KaZa:14, KaZa:16} and \cite{FoHu:16}. We use the notation from \cite{FoSiZa:13}:
 \begin{align}\label{def_dk}
 D_k &= R(t,x;\lambda)^k \partial_x^k, \qquad k = 1,2, \cdots,\\
 \Ltx(\lambda) &= \partial_t + \frac{1}{2}\lambda^2D_2 + \lambda^2D_1.\label{def_ltx}
 \end{align}
 Note that the coefficients of $\Ltx(\lambda)$ depend on $R(t,x;\lambda)$, and therefore on $M(t,x;\lambda)$. The Merton PDE \eqref{eq_value} can be re-written as
 \begin{align}\label{eq_mertonlinear}
 &\Ltx(\lambda)M(t,x;\lambda) = 0.
 \end{align}
 
 Next, we summarize all assumptions needed in the rest of this section. This will include properties of the utility function $U(x)$,  the state processes $(X_t^\pz, S_t, \Zh{t})$ as well as $\vz(t,x,z)$.
 
 \subsection{Assumptions}\label{sec_assumpU}
  Basically, we work under the same set of assumptions as in \cite{FoHu:16}, and we restate them here for readers' convenience.  Detailed discussion about general utility functions can be found there in Section~2.3.
 \begin{assump}\label{assump_U}
 	Throughout the paper, we make the following assumptions on the utility $U(x)$:
 	\begin{enumerate}[(i)]
 		\item\label{assump_Uregularity}  U(x) is $C^6(0,\infty)$, strictly increasing, strictly concave and satisfying the following conditions (Inada and Asymptotic Elasticity):
 		\begin{equation}\label{eq_usualconditions}
 		U'(0+) = \infty, \quad U'(\infty) = 0, \quad \text{AE}[U] := \lim_{x\rightarrow \infty} x\frac{U'(x)}{U(x)} <1.
 		\end{equation}
 		\item\label{assump_Ubddbelow}U(0+) is finite. Without loss of generality, we assume U(0+) = 0.
 		\item\label{assump_Urisktolerance} Denote by $R(x)$ the risk tolerance, 
 		\begin{equation}\label{eq_risktolerance}
 		R(x) := -\frac{U'(x)}{U''(x)}.
 		\end{equation}
 		Assume that $R(0) = 0$, R(x) is strictly increasing and $R'(x) < \infty$ on $[0,\infty)$, and there exists $K\in\RR^+$, such that for $x \geq 0$, and $ 2\leq i \leq 4$,
 		\begin{equation}\label{assump_Uiii}
 		\abs{\partial_x^{(i)}R^i(x)} \leq K.
 		\end{equation}
 		\item\label{assump_Ugrowth} Define the inverse function of the marginal utility $U'(x)$ as $I: \RR^+ \to \RR^+$, $I(y) = U'^{(-1)}(y)$, and assume that, for some positive $\alpha$, $\kappa$, $I(y)$ satisfies the polynomial growth condition:
 		\begin{equation}\label{cond_I}
 		I(y) \leq \alpha + \kappa y^{-\alpha}, 
 		\end{equation}
 		as well as for positive constants 	$c_n$, $C_n$, $n = 1, 2, 3$, with $c_2 >1$,
 		\begin{equation}\label{cond_Iprime}
 		 		c_1I(x) \leq \abs{xI'(x)} \leq C_1 I(x), \quad c_2 \abs{I'(x)} \leq xI''(x) \leq C_2 \abs{I'(x)} \text{ and } \abs{xI'''(x)} \leq C_3 I''(x),
 		\end{equation}
 	
 	\end{enumerate}
 \end{assump}
 
 \begin{rem}\label{rem_U}
 The item (ii) excludes the case of power utility $U(x) = \frac{x^{1-\gamma}}{1-\gamma}$ when $\gamma > 1$. However, all results in this section still hold for the case $\gamma >1$, with a slight modification in the proofs.
 
 The conditions \eqref{cond_Iprime} which were introduced in \cite{KaZa:16}, are crucial assumptions in their Proposition 4, which will be used in our derivation. They also give a mixture of inverse of the marginal utilities as an example that satisfies this condition.
 
 Under condition \eqref{cond_I}, the risk-tolerance function $R(t,x;\lambda)$ is smooth in the variable $\lambda$. This property will be used in the derivation of Proposition \ref{prop_Vz}.
 	 To prove it, we see from \cite[Proposition 3.3(iii)]{FoHu:16} that the risk-tolerance function $R(t,x;\lambda)$ can be expressed by 
 \begin{equation*}
 R(t,x;\lambda) = H_x(H^{(-1)}(x,t,\lambda),t,\lambda),
 \end{equation*}
 where $H(x,t;\lambda): \RR\times [0,T] \times \RR \rightarrow \RR^+$ is the unique solution to the heat equation
 \begin{equation*}
 H_t + \half \lambda^2 H_{xx} = 0, \quad H(x,T,\lambda) = I(e^{-x}),
 \end{equation*}
 and $H^{(-1)}$ is the inverse function of the variable $x$. Then it follows by the fact that $H(x,t;\lambda)$ is smooth in the parameter $\lambda$. 

 \end{rem}

 Below are the additional assumptions needed on the state processes $(X_t^\pz, S_t, Z_t)$ and on $\vz(t,x,z)$.
 
 \begin{assump}\label{assump_valuefunc} %We make the following assumptions on the state 
 	\quad 
 	\begin{enumerate}[(i)]
 		\item\label{assump_lambda} The function $\lambda(z) = \mu(z)/\sigma(z)$ is  $C^2(\RR)$. Moreover, $\lambda(z)$,  $\lambda'(z)$ and $\lambda''(z)$ are at most polynomially growing.
 		
 		\item \label{assump_vxx}The value function $\vz(t,x,z) = M(t,x;\lambda(z))$ satisfies the relation: 
 		\begin{equation}\label{eq_vzvzxx}
 		\abs{x^2\vz_{xx}(t,x,z)} \leq d(z)\vz(t,x,z),
 		\end{equation}
 		with $d(z)$ being of polynomial growth. Note that this is automatically satisfied by the power utility \eqref{def_power}.
 		
 		\item\label{assump_vz}  The process $\vz(t,X_t^\pz, \Zh{0})$ is in $L^4([0,T]\times \Omega)$ uniformly in $\delta$, i.e.,
 		\begin{equation}\label{eq_vz}
 		\EE\left[\int_0^T \left(\vz(s,X_s^\pz,\Zh{0})\right)^4 \ud s\right] \leq C_1
 		\end{equation}
 		where $C_1$ is independent of $\delta$ and $\Zh{0}$ is given in \eqref{def_Zh} with $t=0$.
 	\end{enumerate}
 \end{assump}
 
 \begin{rem}
 	Notice that condition \eqref{eq_vzvzxx} is actually a hidden assumption on the general utility, and it is automatically satisfied by power utility. In order to guarantee \eqref{eq_vz}, there is a list of assumptions discussed in \cite[Section~2.4]{FoHu:16}.
 \end{rem}

\subsection{The Epsilon-Martingale Decomposition with a Given Strategy $\pz$}\label{sec_decompz}

As introduced in \cite{FoPaSi:00} in the context of linear pricing problem and further developed in \cite{GaSo:15}, the idea of epsilon-martingale decomposition is to find a process which is in the form of a martingale plus something small with the right terminal condition. Specifically, we aim to find $Q^{\pz,\delta}$ such that its terminal condition coincides with the quantity of interest $\Vzl_t$, namely, $Q_T^{\pz, \delta}= \Vzl _T= U(X_T^\pz)$, and that can be decomposed as
\begin{equation} \label{eq_epsmart}
Q_t^{\pz, \delta} = M_t^\delta + R_t^\delta,
\end{equation}
where $M_t^\delta$ is a martingale and $R_t^\delta$ is of order  $\delta^{2H}$. Note that the term of order $\delta^H$ will be absorbed in the martingale $M_t^\delta$.

Suppose we obtain such a decomposition \eqref{eq_epsmart}, and then taking conditional expectation with respect to $\MCF_t$ on both sides of the equation $Q_T^{\pz, \delta} = M_T^\delta + R_T^\delta$ gives
\begin{equation}
\Vzl_t = \EE\left[Q_T^{\pz, \delta} \vert \MCF_t\right] = M_t^\delta + \EE\left[R_T^\delta \vert \MCF_t \right] = Q_t^{\pz, \delta}  + \EE\left[R_T^\delta \vert \MCF_t \right] - R_t^\delta.
\end{equation}
Since $R_t^\delta$ is of order $\delta^{2H}$, $Q_t^{\pz, \delta}$ is the approximation to $\Vzl_t$  up to $\delta^H$. Therefore the above argument leads to the desired approximation result. Now it remains to find $Q_t^{\pz, \delta}$ so that the decomposition holds, and we have the following proposition.

\begin{prop}\label{prop_Vz}
Under Assumption~\ref{assump_U} and \ref{assump_valuefunc}, for fixed $t \in [0,T)$, $X_t^\pz = x$, and the observed value $\Zh{0}$, the $\MCF_t$-measurable value process $\Vzl_t$ defined in \eqref{def_Vz} is of the form
\begin{equation}\label{eq_Vtexpansion}
\Vzl_t = Q_t^{\pz, \delta}(X_t^\pz, \Zh{0}) + \MCO(\delta^{2H}),
\end{equation}
	where $Q_t^{\pz, \delta}(x,z)$ is given by:
	\begin{equation}\label{def_Qt}
	Q_t^{\pz, \delta}(x,z) = \vz(t,x,z) + \lambda(z)\lambda'(z)D_1\vz(t,x,z) \phi_t^\delta + \delta^H \rho\lambda^2(z)\lambda'(z)\vo(t,x,z),
	\end{equation}
	$\vz$ and $D_1$ are defined in \eqref{def_vz} and \eqref{def_dk} respectively, $\left(\phi_t^\delta\right)_{t \in [0,T]}$ is the $\MCF_t$-measurable process of order $\delta^{H}$ given in \eqref{def_phi}
%	\begin{equation}\label{def_phi}
%	\phi_t^\delta = \EE\left[\left.\int_t^T \Zh{s} - \Zh{0} \ud s \right\vert \MCF_t\right],
%	\end{equation}
	and $\vo(t,x,z)$ is defined as
	\begin{equation}\label{def_vodt}
	\vo(t,x,z) = D_1^2\vz(t,x,z) D_{t,T},  \quad D_{t,T} = \frac{(T-t)^{H+3/2}}{\Gamma(H + \frac{5}{2})}.
	\end{equation}

%where $Q^\delta(t,x,z)$  is given by \eqref{def_Qt}, $X_t^\pz$ follows \eqref{def_Xpz} and $\Zh{0}$ is defined in \eqref{def_Zh} with $t = 0$.
\end{prop}

The proof of Proposition~\ref{prop_Vz} will be given after Corollary~\ref{cor_power} and Proposition~\ref{prop_markovian}. As explained in Remark~\ref{rem_piapprox}, we have the following corollary.
\begin{cor}\label{cor_power}
In the case of power utility $U(x) = \frac{x^{1-\gamma}}{1-\gamma}$, with $\gamma >0$, $\gamma \neq 1$,
and under Assumption \ref{assump_power} and \ref{assump_valuefunc}, $\pz$ given by \eqref{eq_piapprox} is asymptotically optimal in the full class of admissible strategies $\MCA_t^\delta$ up to order $\delta^H$.
\end{cor}
\begin{proof}
Straightforward computations give, under power utilities, $\vz$, $D_1\vz$ and $\vo$ as
	\begin{align*}
	&\vz(t,x,z) = \frac{x^{1-\gamma}}{1-\gamma}e^{\frac{1-\gamma}{2\gamma}\lambda^2(z)(T-t)}, \quad D_1\vz(t,x,z) =  \frac{x^{1-\gamma}}{\gamma}e^{\frac{1-\gamma}{2\gamma}\lambda^2(z)(T-t)},\\
	& \vo(t,x,z) = \frac{(1-\gamma)}{\gamma^2}x^{1-\gamma}e^{\frac{1-\gamma}{2\gamma}\lambda^2(z)(T-t)} \frac{(T-t)^{H+\frac{3}{2}}}{\Gamma(H + \frac{5}{2})}.
	\end{align*}
Then, one can deduce $Q_t^\delta = Q_t^{\pz,\delta}$, where $Q_t^\delta$ is given by \eqref{def_Qdelta} and $Q_t^{\pz, \delta}$ is given by \eqref{def_Qt}. Combining with Theorem~\ref{thm_Vtpowerexpansion}, $\Vz$ and $\Vzl$ admits the same first order approximation. Therefore, we obtain the desired asymptotic optimality.
%	Consequently, following Proposition~\ref{prop_Vz}, we obtain that the  \emph{zeroth order strategy} $\pi_t^{(0)} = \frac{\lambda(\Zh{t})}{\gamma \sigma(\Zh{t})}$ generate the approximation \eqref{eq_Vtpower}, and therefore $\pi_t^{(0)}$ is asymptotically optimal in the full class of trading strategies $\MCA_t$ up to order $\delta^H$.
\end{proof}

For general utilities, we will derive a similar result in the smaller class $\widetilde \MCA_t^\delta[\pzt, \pot, \alpha]$ in Section~\ref{sec_optimalitypz}.

\begin{prop}\label{prop_markovian}
	For the Markovian case $H= \half$,	the approximation $Q_t^{\pz,\delta}$ given in \eqref{eq_Vtexpansion} coincides with the result derived in \cite[Theorem 3.1]{FoHu:16},
	\begin{equation}
	\Vzl(t,x,z) = \vz(t,x,z) + \frac{\sqrt{\delta}}{2}(T-t)^2\rho\lambda^2(z)\lambda'(z)D_1^2\vz(t,x,z) + \MCO(\delta).
	\end{equation}
\end{prop}

\begin{proof}	
	First observe that when $H = \half$, 
	\begin{equation*}
	 D_{t,T} = \half(T-t)^2,
	\end{equation*}
	and the third term in $Q_t^{\pz, \delta}$ becomes $\frac{\sqrt{\delta}}{2}(T-t)^2\rho \lambda^2(z)\lambda'(z) D_1^2\vz(t,x,z)$. Using the moving-average representation \eqref{def_Zh} for $\Zh{s}$ with $H = 1/2$, $\phi_t^\delta$ is explicitly computed as
	\begin{equation*}
		\phi_t^\delta =  \frac{1-e^{-a\delta(T-t)}}{a\delta} \Zh{t} - (T-t)\Zh{0}= (T-t)\left(\Zh{t}-\Zh{0}\right) + \MCO(\delta).
	\end{equation*}
Then using the ``Vega-Gamma'' relation $\vz_z(t,x,z) = (T-t)\lambda(z)\lambda'(z)D_1\vz(t,x,z)$ and the fact \\
$\left(\Zh{t} - \Zh{0}\right)^p \sim \MCO(\delta^{pH})$, one can deduce
	\begin{align*}
	\Vzl_t =& \vz(t,X_t^\pz, \Zh{0}) + \vz_z(t,X_t^\pz, \Zh{0})\left(\Zh{t} - \Zh{0}\right) \\
	&+ \frac{\sqrt{\delta}}{2}(T-t)^2\rho \lambda^2(\Zh{0})\lambda'(\Zh{0}) D_1^2\vz(t,X_t^\pz, \Zh{0})\\
	 =& \vz(t,X_t^\pz, \Zh{t}) + \frac{\sqrt{\delta}}{2}(T-t)^2\rho \lambda^2(\Zh{t})\lambda'(\Zh{t}) D_1^2\vz(t,X_t^\pz, \Zh{t}) + \MCO(\delta),
	\end{align*}
which is consistent with the result derived in \cite[Theorem~3.1]{FoHu:16}.
\end{proof}

We now turn to the proof of Proposition~\ref{prop_Vz}.
\begin{proof}[Proof of Proposition~\ref{prop_Vz}]
According to the epsilon-martingale decomposition strategy, our goal is to show that $Q^{\pz, \delta}_t$ can be written as $M_t^\delta + R_t^\delta$, where $M_t^\delta$ is a martingale, and $R_t^\delta$ is of order $\delta^{2H}$. We shall mainly focus on the derivation of $Q_t^{\pz,\delta}$ and delay the proofs of accuracy in the Appendix \ref{app_lemmas} for the sake of clarity and simplicity. The technique is very similar to the one presented in \cite{GaSo:15} in the context of option pricing problem with fractional stochastic volatility. The main difference is that their case involves the linear Black-Scholes operator, as in our case, it involves the non-linear Merton operator $\Ltx(\lambda)$. Amazingly, the properties of risk-tolerance function $R(t,x;\lambda)$ will enable us to carry out the proof as follows.

In order to avoid differentiating the fOU process $\Zh{t}$, we freeze it at $\Zh{0}$, and the corresponding error will be compensated in the following calculation. This technique has also been used in the context of pricing when deriving hedging strategy with frozen volatility in \cite[Section~8.4]{FoPaSiSo:11}.

By It$\hat{\text{o}}$'s formula applied to $\vz$ defined in \eqref{def_vz} and by a Taylor expansion in $z$ at the point $\Zh{0}$, we deduce 
\begin{align}
\ud \vz(t,X_t^\pz,\Zh{0}) &= \Ltx(\lambda(\Zh{t}))\vz(t,X_t^\pz, \Zh{0}) \ud t \nonumber\\
& \hspace{10pt} + \sigma(\Zh{t})\pz(t,X_t^\pz, \Zh{t})\vz_x(t,X_t^\pz, \Zh{0})\ud W_t \nonumber\\
& = \Ltx(\lambda(\Zh{0}))\vz(t,X_t^\pz, \Zh{0}) \ud t \nonumber \\
&\hspace{10pt} + \left[(\Zh{t}-\Zh{0})(\lambda^2R)_z\big\vert_{z = \Zh{0}} + g^{(1)}_t\right]\vz_x(t,X_t^\pz, \Zh{0}) \ud t \nonumber\\
& \hspace{10pt} + \left[(\Zh{t}-\Zh{0})(\lambda^2R^2)_z\big\vert_{z=\Zh{0}} + g_t^{(2)}\right]\half\vz_{xx}(t,X_t^\pz,\Zh{0}) \ud t + \ud M_t^{(1)} \nonumber\\
& = (\Zh{t}-\Zh{0})\lambda(\Zh{0})\lambda'(\Zh{0})D_1\vz(t, X_t^\pz, \Zh{0}) \ud t + \ud M_t^{(1)} \nonumber\\
& \hspace{10pt} + g_t^{(1)}\vz_x(t,X_t^\pz,  \Zh{0}) \ud t + \half g_t^{(2)} \vz_{xx}(t,X_t^\pz, \Zh{0}) \ud t,  \label{eq_g1g2}
\end{align}
where in the derivation, we have used the relation
\begin{equation}\label{eq_relation}
\Ltx(\lambda(z))\vz(t,x,z) = 0, \quad D_1\vz = -D_2\vz, \text{ and } \pz(t,x,z) = \frac{\lambda(z)}{\sigma(z)}R(t,x;\lambda(z)),
\end{equation}
$M_t^{(1)}$ is the martingale defined by
\begin{equation}\label{def_m1}
\ud M_t^{(1)} = \sigma(\Zh{t})\pz(t,X_t^\pz, \Zh{t})\vz_x(t,X_t^\pz, \Zh{0})\ud W_t, 
\end{equation}
and the last two terms in \eqref{eq_g1g2} are of order $\delta^{2H}$ (see Appendix \ref{app_lemmas}), with 
$g_t^{(1)}$ and $g_t^{(2)}$ being the Lagrange remainders:
\begin{align}\label{def_g1}
g_t^{(1)} = \half\left(\Zh{t}-\Zh{0}\right)^2\left(\lambda^2 R\right)_{zz}\Big\vert_{z = \chi_t^{(1)}}, \quad g_t^{(2)} =  \half\left(\Zh{t}-\Zh{0}\right)^2\left(\lambda^2 R^2\right)_{zz}\Big\vert_{z = \chi_t^{(2)}},
\end{align}
and $ \chi_t^{(i)} \in \left[\Zh{0}\wedge\Zh{t}, \Zh{0}\vee\Zh{t}\right]$, $ i = 1, 2$.

Now it remains to find the epsilon-martingale decomposition for the term \\
 $\int(\Zh{s}-\Zh{0})D_1\vz(s, X_s^\pz, \Zh{0}) \ud s$ in \eqref{eq_g1g2}. To this end, we recall $\phi_t^\delta$ and $\psi_t^\delta$ given in \eqref{def_phi} and \eqref{def_psi} respectively, 
%and define the martingale $\psi_t^\delta$ by
%\begin{equation}\label{def_psi}
%\psi_t^\delta = \EE\left[\left.\int_0^T \Zh{s}-\Zh{0}\ud s \right\vert \MCF_t\right].
%\end{equation}
which satisfy the relation $\left(\Zh{t}-\Zh{0}\right) \ud t = \ud \psi_t^\delta - \ud \phi_t^\delta$ and consequently
\begin{equation}\label{eq_ztz0}
\left(\Zh{t}-\Zh{0}\right)D_1\vz(t,X_t^\pz, \Zh{0}) \ud t = D_1\vz(t,X_t^\pz, \Zh{0}) \left(\ud \psi_t^\delta - \ud \phi_t^\delta\right).
\end{equation}
On the right-hand side, the first term is proved to be a true martingale in Appendix \ref{app_lemmas}, while the second term need further analysis, namely, the differential of $\phi_t^\delta D_1\vz$ will be computed. In the sequel, 
without any confusion, the arguments of $\vz(t,X_t^\pz, \Zh{0})$ shall be omitted for simplicity.
\begin{align}
\ud \left(\phi_t^\delta D_1\vz\right) = &D_1\vz\ud \phi_t^\delta + \phi_t^\delta \Ltx(\lambda(\Zh{t}))D_1\vz \ud t + \phi_t^\delta \sigma(\Zh{t})\pz(t,X_t^\pz, \Zh{t})\partial_xD_1\vz \ud W_t\nonumber\\ 
&+\sigma(\Zh{t})\pz(t,X_t^\pz,\Zh{t})\partial_xD_1\vz \ud \average{W,\phi^\delta}_t\nonumber \\
 =& D_1\vz\ud \phi_t^\delta  + \rho\lambda(\Zh{0})D_1^2\vz \ud \average{W^Z,\psi^\delta}_t +  \phi_t^\delta \sigma(\Zh{t})\pz(t,X_t^\pz, \Zh{t})\partial_xD_1\vz \ud W_t\nonumber\\
&+  \phi_t^\delta g_t^{(3)} \partial_xD_1\vz \ud t + \phi_t^\delta g_t^{(4)} \half\partial_{xx} D_1\vz \ud t + \rho g_t^{(5)}\partial_xD_1\vz \ud \average{W^Z,\psi^\delta}_t ,\label{eq_g3g4g5}
\end{align}
where in the above derivation, we have used 
\begin{equation}
\Ltx(\lambda(\Zh{0}))D_1\vz = D_1\Ltx(\lambda(\Zh{0}))\vz = 0, \text{ and }
\ud \average{W,\phi^\delta}_t = \rho \ud \average{W^Z, \psi^\delta}_t,
\end{equation}
with the first one being proved in \cite[Lemma~2.5]{FoSiZa:13}.
Again, $g_t^{(3)}$, $g_t^{(4)}$ and $g_t^{(5)}$ are Lagrange remainders from Taylor series
\begin{align*}
&g_t^{(3)} = \left(\Zh{t}-\Zh{0}\right)(\lambda^2R)_z\Big\vert_{z = \chi_t^{(3)}}, \quad  
g_t^{(4)} =  \left(\Zh{t}-\Zh{0}\right)\left(\lambda^2 R^2\right)_{z}\Big\vert_{z = \chi_t^{(4)}},\\ &g_t^{(5)} =  \left(\Zh{t}-\Zh{0}\right)\left(\lambda R\right)_{z}\Big\vert_{z = \chi_t^{(5)}}, 
\end{align*}
with $\chi_t^{(i)} \in \left[\Zh{0}\wedge\Zh{t}, \Zh{0}\vee\Zh{t}\right]$, $i = 3, 4, 5$.

Now combining \eqref{eq_ztz0} and \eqref{eq_g3g4g5} yields:
\begin{align*}
\left(\Zh{t}-\Zh{0}\right)D_1\vz \ud t &= -\ud \left(\phi_t^\delta D_1\vz \right) + \rho\lambda(\Zh{0})D_1^2\vz\ud \average{W^z,\psi^\delta}_t + \ud M_t^{(2)}\\
&\hspace{10pt}+ \phi_t^\delta g_t^{(3)} \partial_xD_1\vz \ud t + \phi_t^\delta g_t^{(4)} \half\partial_{xx} D_1\vz \ud t + \rho g_t^{(5)}\partial_xD_1\vz \ud \average{W^Z,\psi^\delta}_t
\end{align*}
where $M_t^{(2)}$ is the martingale given by
\begin{equation}\label{def_m2}
\ud M_t^{(2)} = D_1\vz(t,X_t^\pz, \Zh{0})\ud \psi_t^\delta + \phi_t^\delta \sigma(\Zh{t})\pz(t,X_t^\pz, \Zh{t})\partial_xD_1\vz(t,X_t^\pz, \Zh{0})  \ud W_t,
\end{equation} 
following a similar proof as for $M_t^{(1)}$.

Let $ \ud \average{W^Z, \psi^\delta}_t := \theta_{t,T}^\delta\ud t$, from Lemma~\ref{lem_moments}\eqref{lem_psi}, one has 
\begin{align*}
 \theta_{t,T}^\delta= \int_{0}^{T-t} \MCK^\delta(s)\ud s = \delta^H \theta_{t,T} + \delta^{H+1}\widetilde\theta_{t,T},
\end{align*}
and a straightforward computation gives
\begin{equation}
 \partial_t D_{t,T} = - \theta_{t,T},
\end{equation}
where $D_{t,T}$ is defined in \eqref{def_vodt}.
Then applying It\^{o}'s formula to $\vo$ defined in \eqref{def_vodt} brings
\begin{align}
\ud \vo(t,X_t^\pz, \Zh{0}) &= \Ltx(\lambda(\Zh{t}))\vo \ud t + \sigma(\Zh{t})\pz(t,X_t^\pz, \Zh{t}) \vo_x \ud W_t \nonumber\\
& = \Ltx(\lambda(\Zh{0}))\vo \ud t + \sigma(\Zh{t})\pz(t,X_t^\pz, \Zh{t}) \vo_x \ud W_t \nonumber\\
&\hspace{10pt}+ g_t^{(3)} \vo_x \ud t + g_t^{(4)} \half\vo_{xx} \ud t \nonumber\\
& = -D_1^2\vz \theta_{t,T} \ud t +  \ud M_t^{(3)} + g_t^{(3)} \vo_x \ud t + g_t^{(4)} \half\vo_{xx} \ud t,\label{eq_vo}
\end{align}
with the last two terms of order  $\MCO(\delta^{H})$, and $M_t^{(3)}$ as the martingale:
\begin{equation}\label{def_m3}
\ud M_t^{(3)} = \sigma(\Zh{t})\pz(t,X_t^\pz, \Zh{t}) \vo_x(t,X_t^\pz, \Zh{0}) \ud W_t.
\end{equation}

Collecting equation \eqref{eq_g1g2}, \eqref{eq_g3g4g5} and \eqref{eq_vo}, we obtain
\begin{align*}
\ud Q^{\pz,\delta}_t(X_t^\pz, \Zh{0}) &= \ud \left( \vz(t,X_t^\pz,\Zh{0}) + \lambda(\Zh{0})\lambda'(\Zh{0})\phi_t^\delta D_1\vz(t,X_t^\pz, \Zh{0})\right.\\
&\hspace{25pt} \left. + \delta^H \rho \lambda^2(\Zh{0})\lambda'(\Zh{0}) \vo(t,X_t^\pz, \Zh{0}) \right) \\
&= \ud M_t^\delta +  \ud R^\delta_t,
\end{align*}
where $\ud M^\delta_t$  and $\ud R^\delta_t$ are
\begin{align}
&\ud M^\delta_t = \ud M_t^{(1)} + \lambda(\Zh{0})\lambda'(\Zh{0})\ud M_t^{(2)} + \delta^H \rho \lambda^2(\Zh{0})\lambda'(\Zh{0})\ud M_t^{(3)}, \label{def_Mt}\\
&\ud R^\delta_t = g_t^{(1)}\vz_x\ud t + \half g_t^{(2)} \vz_{xx} \ud t +
\delta^H \rho\lambda^2(\Zh{0})\lambda'(\Zh{0})\left[ g_t^{(3)} \vo_x \ud t + g_t^{(4)} \half\vo_{xx} \ud t\right]\label{def_Rt} \\
 & +\lambda(\Zh{0})\lambda'(\Zh{0})\left[\phi_t^\delta g_t^{(3)} \partial_xD_1\vz \ud t + \phi_t^\delta g_t^{(4)} \half\partial_{xx} D_1\vz \ud t + \rho g_t^{(5)}\partial_xD_1\vz \left(\delta^H \theta_{t,T} + \delta^{H+1}\widetilde \theta_{t,T}\right) \ud t \right]. \nonumber
\end{align}

Noticing that $\vz(T,X_T^\pz, \Zh{0}) = U(X_T^\pz)$, $\phi_T^\delta D_1\vz(T,X_T^\pz, \Zh{0}) = 0 $ since $\phi_T^\delta = 0$, and $\vo(T,X_T^\pz, \Zh{0}) =0$ by definition, the terminal condition for $Q^{\pz, \delta}$ indeed coincides with $\Vzl_T$. Combining with the proof that $M_t^\delta$ is a true martingale and $R_t^\delta$ is of order $\delta^{2H}$ detailed in Appendix~\ref{app_lemmas}, we obtain the desired result in Proposition~\ref{prop_Vz}.
\end{proof}

\subsection{Asymptotic Optimality of $\pz$}\label{sec_optimalitypz}

Recall the specific family of admissible strategies $\widetilde{\MCA}_t^\delta$ defined in \eqref{def_Atilde}:
\begin{equation}
\widetilde \MCA_t^\delta[\pzt, \pot, \alpha] := \left\{\pi = \pzt + \delta^\alpha \pot: \pi \in \MCA_t^\delta, \alpha >0, 0< \delta \leq 1 \right\},
\end{equation}
with $\pzt$ and $\pot$ being feedback controls, and $\MCA_t^\delta$ being the set of all admissible strategies defined in \eqref{def_MCA}.
In this subsection, we first derive the approximation of $\Vzp_t$
\begin{equation}\label{def_Vzpi}
\Vzp_t := \EE\left[\left.U\left( X_T^\pi \right)\right\vert \MCF_t\right], 
%\quad \forall \; \pi = \pzt + \delta^\alpha \pot \in \widetilde \MCA_t^\delta,
\end{equation}
for any admissible strategy $\pi \in \widetilde \MCA_t^\delta$ taking the form $\pzt + \delta^\alpha \pot$ using epsilon-martingale decomposition technique as demonstrated in Proposition~\ref{prop_Vz},
%where $\pi =\pzt + \delta^\alpha \pot$ is a admissible strategy in $\widetilde \MCA_t^\delta$ satisfying Assumption~\ref{assump_piregularity}
where $X_t^\pi$ is the wealth process following the trading strategy $\pi$:
\begin{equation}\label{def_Xtilde}
\ud X_t^\pi = \mu(\Zh{t})\pi(t,X_t^\pi, \Zh{t}) \ud t + \sigma(\Zh{t})\pi(t,X_t^\pi, \Zh{t}) \ud W_t.
\end{equation}
Then, given the previously established results of $\Vzl_t$ in Proposition~\ref{prop_Vz}, we asymptotically compare these approximations for $\Vzl_t$ and $\Vzp_t$, and then prove Theorem~\ref{thm_main}.

\begin{theo}\label{thm_main}
	Under Assumptions~\ref{assump_U}, \ref{assump_valuefunc}, \ref{assump_piregularity} and \ref{assump_optimality}, for any family of trading strategies $\widetilde \MCA_t^\delta[\pzt, \pot, \alpha]$, the following limit exists in $L^2$ and satisfies
	\begin{equation}
	\ell := \lim_{\delta \to 0}\frac{\Vzp_t - \Vzl_t}{\delta^H} \leq 0, \text{ in } L^2,
	\end{equation}
	where $\Vzl_t$ and $\Vzp_t$ are defined in \eqref{def_Vz} and \eqref{def_Vzpi} respectively.
	
%	 $V^\pz_t$ and $V^\pi_t$ are given by
%	\begin{equation}
%	V_t^\pz = \EE\left[\left.U(X_T^\pz)\right\vert \MCF_t\right], \quad 
%	V_t^\pi = \EE\left[\left.U(X_T^{\pzt + \delta^\alpha\pot})\right\vert \MCF_t\right],
%	\end{equation}
%	with $X_t^\pz$ and $X_t^{\pzt + \delta^\alpha\pot}$ being the wealth processes \eqref{def_Xt} with $\pi = \pz$ and $\pi = \pzt + \delta^\alpha\pot \in \widetilde \MCA_t^\delta$ respectively.
	
	That is, the strategy $\pz$ that generates $\Vzl_t$ performs asymptotically better up to order $\delta^H$ than any family $\widetilde \MCA_t^\delta[\pzt, \pot, \alpha]$. Moreover, the inequality can be written according to the following four cases:
	\begin{enumerate}[(i)]
		\item $\pzt = \pz$, $\alpha > H/2$: $\ell = 0$ and $\Vzp_t = \Vzl_t + o(\delta^H)$; 
		\item $\pzt = \pz$, $\alpha = H/2$: $-\infty < \ell < 0$ and $\Vzp_t= \Vzl_t + \MCO(\delta^H)$ with $\MCO(\delta^H) <0$;
		\item $\pzt = \pz$, $\alpha < H/2$: $\ell = -\infty$ and $\Vzp_t = \Vzl_t + \MCO(\delta^{2\alpha})$ with $\MCO(\delta^{2\alpha}) <0$;
		\item $\pzt \neq \pz$: $\lim_{\delta \to 0} \Vzp_t < \lim_{\delta \to 0} \Vzl_t$,
	\end{enumerate}
	where all relations between $\Vzp_t$ and $\Vzl_t$ hold under $L^2$ sense.
\end{theo}

%
%
%We need to compare $\Vzt$ with $\Vzl$ defined in  \eqref{def_Vzl}, for which we have established the first order approximation $\vz+\sqrt{\delta}\vo$ in Theorem \ref{Thm_one}. This comparison is asymptotic in $\delta$ up to order $\sqrt{\delta}$, and our first step is to obtain the corresponding approximation for $\Vzt$. This is done heuristically in Section \ref{sec_heuristicVzt} in the two cases $\pzt \equiv \pz$ and $\pzt \not\equiv \pz$, and depending on the value of the parameter $\alpha$.  The proof of accuracy is given in Section \ref{sec_accuracy2}. Asymptotic optimality of $\pz$ is obtained in Section \ref{sec_optimalitypz}.
%

\begin{assump}\label{assump_piregularity}
	For a fixed choice of $(\pzt$, $\pot$, $\alpha>0)$, we require:
	\begin{enumerate}[(i)] 
		\item The whole family (in $\delta$) of strategies $\{\pzt + \delta^\alpha \pot\}$ is contained in $\MCA^\delta_t$;
		\item The function $\mu(z)$ is $C^1(\RR)$.
		\item  Functions $\pzt(t,x,z)$ and $\pot(t,x,z)$ are continuous on $[0,T]\times \RR^+\times \RR$, and $C^1$ in $z$.
		 \item The process $\vz(t,X_t^\pi, \Zh{0})$ is in $L^4([0,T]\times \Omega)$ uniformly in $\delta$, i.e.,
		   \begin{equation}
		   \EE\left[\int_0^T \left(\vz(s,X_s^\pi,\Zh{0})\right)^4 \ud s\right] \leq C_2
		   \end{equation}
		   where $C_2$ is independent of $\delta$, $\Zh{0}$ follows \eqref{def_Zh} with $t=0$, and $X_t^\pi$ follows \eqref{def_Xtilde} with $\pi = \pzt + \delta^\alpha \pot$.
%		\item Let  $(\widetilde{X}_s^{t,x})_{t\leq s\leq T}$ be the solution to:
%		\begin{equation}\label{eq_Xttilde}
%		\ud \widetilde X_s = \mu(z)\pzt(s,\widetilde X_s, z) \ud s + \sigma(z) \pzt(s,\widetilde X_s,z) \ud W_s,
%		\end{equation}
%		starting at $x$ at time $t$. 
%		
%		By (i), $\widetilde{X}_s^{t,x}$ is nonnegative and we further
%		assume that it has full support $\RR^+$ for any $t<s\leq T$.
%	
\end{enumerate}
\end{assump}
%\begin{rem}\label{rem_pztpot}
%	Notice that $\pz$ defined in \eqref{pi0} is continuous on $[0,T]\times \RR^+\times \RR$, thus, it is natural to  require that $\pzt$ and $\pot$ have the same regularity as $\pz$, that is (ii). Regarding (iii),
%	from  Section \ref{sec:assumptions}, $\pz$ is the optimal trading strategy for the Merton problem when $\delta =0$, in which case $Z_t$ is  frozen at its initial position $z$. The associated wealth process $\widehat X_s^{t,x}$  starting at $x$ at time $t$ is the solution to
%	\begin{equation*}
%	\ud \widehat X_s = \mu(z)\pz(s,\widehat X_s, z) \ud s + \sigma(z)\pz(s,\widehat X_s, z) \ud W_s, \quad \widehat X_t = x.
%	\end{equation*}
%	Then, from \cite[Proposition 7]{KaZa:14}, one has	
%	\begin{equation*}
%	\widehat X_s^{t,x} = H\left(H^{-1}(x,t,\lambda(z)) + \lambda^2(z)(s-t) + \lambda(z)(W_s-W_t), s, \lambda(z)\right),
%	\end{equation*}
%	where $H: \RR\times[0,T]\times \RR \to \RR^+$ is defined in Proposition \ref{prop_H} and is of full range. Consequently,  $\widehat X_s^{t,x}$ has full support $\RR^+$, and thus, it is natural to require that $\widetilde X_s^{t,x}$ has full support $\RR^+$, that is (iii). 
%\end{rem}
\begin{rem}
	We have $\pzt + \delta^0 \pot = \pzt + \pot + \delta^\alpha \cdot 0$, so it is enough to consider $\alpha >0$.
\end{rem}

\begin{rem}
To demonstrate the non-restrictiveness of Assumption \eqref{assump_optimality}, we give the following example in the case of power utility. We comment that such a choice of utility functions is for the sake of convenience, while Theorem \ref{thm_main} works in general. 

For case (i), if we choose the admissible strategy $\pi = \pzt + \delta^\alpha\pot$ with $\pzt = \pot = \pz$ (the admissibility can be shown similarly as in Theorem \ref{thm_martdistort}), then we deduce that all quantities that are required to be uniformly bounded in $\delta$ are of the form
\begin{equation*}
\EE \int_0^T \mc{P}(\Zh{0}, \chi_s, \Zh{s}, \phi_s^\delta) \left(X_s^\pi\right)^{2(1-\gamma)} \ud s,\quad \chi_s \in \left[\Zh{0}\wedge\Zh{s}, \Zh{0}\vee \Zh{s}\right],
\end{equation*}
where $\mc{P}$ is at most polynomially growing. By H\"{o}lder inequality, they are less than 
\begin{equation*}
\left(\EE \int_0^T \mc{P}(\Zh{0}, \chi_s, \Zh{s}, \phi_s^\delta)^{q} \ud s\right)^{1/q} \left(\EE \int_0^T \left(X_s^\pi\right)^{2p(1-\gamma)} \ud s\right)^{1/p}, \quad 1/p + 1/q = 1.
\end{equation*}
The first quantity is uniformly bounded in $\delta$ by Lemma \ref{lem_moments}\eqref{lem_Zh}\eqref{lem_phi}, while the boundedness of the second one follows by the admissibility of $\pi = \pzt + \delta^\alpha \pot \in \MCA_t^\delta$. An example of case (ii), with the choice  $\pzt = c\pz$ and $\pot = \pz$ could also be validated in a similar manner.
\end{rem}

\begin{proof}
We first deal with the case $\pi = \pz + \delta^\alpha\pot$. The derivation is similar to the one in Section~\ref{sec_decompz}. As usual, in order to condense the notation,  we systematically omit  the argument $(s,X_s^\pi,\Zh{0})$ for $\vz$ in what follows.
%By applying the It\^{o}'s formula to $\vz(t,X_t^\pi, \Zh{0})$ and Taylor expansion at $z = \Zh{0}$,
\begin{align*}
\ud \vz(t,X_t^\pi,\Zh{0}) =& \vz_t\ud t+ \mu(\Zh{t})\pi(t,X_t^\pi,\Zh{t})\vz_x\ud t + \frac{1}{2}\sigma^2(\Zh{t})\pi^2(t,X_t^\pi,\Zh{t})\vz_{xx} \ud t \\
& +\sigma(\Zh{t})\pi(t,X_t^\pi, \Zh{t})\vz_{x}\ud W_t \nonumber\\
 =& (\Zh{t}-\Zh{0})\lambda(\Zh{0})\lambda'(\Zh{0})D_1\vz\ud t + g_t^{(1)}\vz_x \ud t + \half g_t^{(2)} \vz_{xx} \ud t + \ud \widetilde M_t^{(1)}\\
&+ \delta^\alpha\widetilde g_t^{(1)}\vz_x \ud t + \delta^\alpha\widetilde g_t^{(2)} \vz_{xx} \ud t  + \half\delta^{2\alpha}\sigma^2(\Zh{t})\left(\pot\right)^2(t,X_t^\pz, \Zh{t}) \vz_{xx} \ud t,
\end{align*}
where $\widetilde M_t^{(1)}$, $\widetilde g_t^{(1)}$ and $\widetilde g_t^{(2)}$ are defined by
\begin{align*}
&\ud \widetilde M_t^{(1)} = \sigma(\Zh{t})\left(\pz(t,X_t^\pi, \Zh{t}) + \delta^\alpha \pot(t,X_t^\pi, \Zh{t})\right)\vz_x(t,X_t^\pi, \Zh{0}) \ud W_t,\\
&\widetilde g_t^{(1)} = \left(\Zh{t}-\Zh{0}\right)(\mu \pot)_z\Big\vert_{z = \widetilde \chi_t^{(1)}}, \quad \widetilde g_t^{(2)} = \left(\Zh{t}-\Zh{0}\right)(\mu R\pot)_z\Big\vert_{z = \widetilde \chi_t^{(2)}},
\end{align*}
with $\widetilde \chi_t^{(i)}  \in \left[\Zh{0}\wedge\Zh{t}, \Zh{0}\vee\Zh{t}\right]$, $i = 1,2$.

Then it suffices to find the epsilon-martingale decomposition for the term \\
$ (\Zh{t}-\Zh{0})D_1\vz(t,X_t^\pi, \Zh{0})\ud t$. Following a similar derivation as in Section~\ref{sec_decompz}, one can deduce
\begin{align*}
\ud Q_t^{\pz,\delta}(X_t^\pi, \Zh{0}) &= \ud \left( \vz(t,X_t^\pi,\Zh{0}) + \lambda(\Zh{0})\lambda'(\Zh{0})\phi_t^\delta D_1\vz(t,X_t^\pi, \Zh{0})\right.\\
&\hspace{25pt} \left. + \delta^H \rho \lambda^2(\Zh{0})\lambda'(\Zh{0}) \vo(t,X_t^\pi, \Zh{0}) \right) \\
&=\ud \widetilde M_t^\delta + \ud \widetilde R_t^\delta + \delta^{2\alpha} \ud N_t^\delta,
\end{align*}
where 
\begin{align*}
&\ud \widetilde M_t^\delta = \ud \widetilde M_t^{(1)} + \lambda(\Zh{0})\lambda'(\Zh{0})\ud \widetilde M_t^{(2)} + \delta^H \rho\lambda^2(\Zh{0})\lambda'(\Zh{0}) \ud \widetilde M_t^{(3)},\\
&\ud \widetilde M_t^{(2)} = D_1\vz(t,X_t^\pi, \Zh{0}) \ud \psi_t^\delta + \phi_t^\delta \sigma(\Zh{t})\pi(t,X_t^\pi,\Zh{t}) \partial_x D_1\vz(t,X_t^\pi, \Zh{0})  \ud W_t,\\
&\ud \widetilde M_t^{(3)} = \sigma(\Zh{t})\pi(t,X_t^\pi, \Zh{t})\vo_x(t,X_t^\pi, \Zh{0}) \ud W_t,\\
&\ud \widetilde R^\delta_t = g_t^{(1)}\vz_x\ud t + \half g_t^{(2)} \vz_{xx} \ud t +  \delta^\alpha\widetilde g_t^{(1)}\vz_x \ud t + \delta^\alpha \widetilde g_t^{(2)} \vz_{xx} \ud t + 
\delta^H \rho\lambda^2(\Zh{0})\lambda'(\Zh{0})\left[ g_t^{(3)} \vo_x + g_t^{(4)} \half\vo_{xx} \right.\\
&\hspace{30pt}\left.+ \delta^\alpha \mu \pot \vo_x + \delta^\alpha \sigma^2\pz\pot\vo_{xx} + \half\delta^{2\alpha}\sigma^2\left(\pot\right)^2\vo_{xx}\right]\ud t + \lambda(\Zh{0})\lambda'(\Zh{0})\phi_t^\delta\left[g_t^{(3)} \partial_xD_1\vz  \right.\\
&\hspace{30pt} \left. + \half g_t^{(4)} \partial_{xx} D_1\vz+\delta^\alpha \mu\pot \partial_x D_1\vz + \delta^\alpha \sigma^2\pz\pot \partial_{xx}D_1\vz + \half \delta^{2\alpha}\sigma^2\left(\pot\right)^2\partial_{xx}D_1\vz\right]\ud t \\
&\hspace{30pt}+ \rho\lambda(\Zh{0})\lambda'(\Zh{0}) \left[ g_t^{(5)}\partial_xD_1\vz \left(\delta^H \theta_{t,T} + \delta^{H+1}\widetilde \theta_{t,T}\right)+ \delta^\alpha \sigma\pot \partial_x D_1\vz \left(\delta^H \theta_{t,T} + \delta^{H+1}\widetilde \theta_{t,T}\right) \right]\ud t , \nonumber\\
& \ud \widetilde N_t^\delta = \half \sigma^2(\Zh{t})\left(\pot(t,X_t^\pi, \Zh{t})\right)^2\vz_{xx}(t,X_t^\pi, \Zh{0}) \ud t.
\end{align*}
To condense the expression for $R_t^\delta$, we omit the arguments for functions $\vz(t,X_t^\pi,\Zh{0})$, $\vo(t,X_t^\pi, \Zh{0})$, $\mu(\Zh{t})$,  $\sigma(\Zh{t})$, $\pz(t,X_t^\pi, \Zh{t})$ and $\pot(t,X_t^\pi,\Zh{t})$.

Since the Merton value $M(t,x;\lambda)$ is strictly concave, so does $\vz(t,x,z) = M(t,x;\lambda(z))$, which implies that $N_t$ is non-increasing. Moreover, under Assumption~\ref{assump_piregularity}, \ref{assump_optimality}, one can prove $\widetilde M_t^\delta$ is a true martingale and $\widetilde R_t^\delta$ is of order $\delta^{H + H \wedge \alpha}$, which yields
\begin{align}
\Vzp_t &= \EE\left[Q_T^{\pz,\delta}\vert \MCF_t\right] = \widetilde M_t^\delta + \EE\left[\widetilde R_T^\delta + \delta^{2\alpha}N_T^\delta \vert \MCF_t\right] \nonumber \\
&= Q_t^{\pz,\delta}( X_t^\pi, \Zh{0}) + \EE\left[\widetilde R_T^\delta - \widetilde R_t^\delta \vert \MCF_t \right] + \delta^{2\alpha}\EE\left[N_T^\delta - N_t^\delta \vert \MCF_t\right] \nonumber\\
& = Q_t^{\pz, \delta}(X_t^\pi, \Zh{0}) + \delta^{2\alpha}\EE\left[N_T^\delta - N_t^\delta \vert \MCF_t\right] + \MCO(\delta^{H+H\wedge \alpha}) \leq Q_t^{\pz, \delta}( X_t^\pi, \Zh{0}) + \MCO(\delta^{H + H \wedge \alpha}),\label{eq_Vexpansionpz}
\end{align}
where in the derivation we have used $\widetilde M_t^\delta + \widetilde R_t^\delta + N_t^\delta = Q^{\pz,\delta}_t(X_t^\pi, \Zh{0})$ and the decreasing property of $N_t$.

\bigskip
The second case is $\pi = \pzt + \delta^\alpha\pot$ with $\pzt \not \equiv \pz$. Here the wealth process $X_t^\pi$ follows
\begin{align}
	\ud X_t^\pi = \mu(\Zh{t})\left(\pzt + \delta^\alpha \pot \right)(t,X_t^\pi, \Zh{t}) \ud t + \sigma(\Zh{t})\left(\pzt + \delta^\alpha \pot\right)(t, X_t^\pi, \Zh{t}) \ud W_t.
\end{align}
Under similar derivations, one can deduce
\begin{align}
\ud \vz(t,X_t^\pi, \Zh{0}) = \ud \widehat M_t^\delta + \ud \widehat R_t^\delta + \ud \widehat N_t^\delta
\end{align}
where 
\begin{align}
&\ud \widehat M_t^\delta = \sigma(\Zh{t})\pi(t,X_t^\pi, \Zh{t})\vz_x(t,X_t^\pi,\Zh{0}) \ud W_t,\\
&\ud \widehat R_t^\delta = \left[\widehat g_t^{(1)}\vz_x + \half \widehat g_t^{(2)}\vz_{xx}\right]\ud t +  \delta^\alpha \left[\mu\pot\vz_x + \sigma^2\pzt\pot\vz_{xx} + \half \delta^\alpha\sigma^2\left(\pot\right)^2\vz_{xx} \right]\ud t,\\ 
&\ud \widehat N_t^\delta =
\half \sigma^2(\Zh{0})\left(\pzt - \pz\right)^2(t,X_t^\pi, \Zh{0})\vz_{xx}(t,X_t^\pi,\Zh{0}) \ud t, 
\end{align}
with $\widehat g_t^{(1)}$ and $\widehat g_t^{(2)}$ defined as
\begin{equation}
\widehat g_t^{(1)} = \left(\Zh{t} - \Zh{0}\right)(\mu \pzt)_z\Big\vert_{z = \widehat\chi_t^{(1)}}, \quad
\widehat g_t^{(2)} = \left(\Zh{t} - \Zh{0}\right)(\sigma^2 \left(\pzt\right)^2)_z\Big\vert_{z = \widehat\chi_t^{(2)}}, 
\end{equation}
and $\widehat \chi_t^{(i)}  \in \left[\Zh{0}\wedge\Zh{t}, \Zh{0}\vee\Zh{t}\right]$, $i = 1,2$.

Here $\widehat N_t^\delta$ is strictly decreasing due to the strict concavity of $\vz$. Under Assumption~\ref{assump_piregularity}, \ref{assump_optimality}, $\widehat M_t^\delta$ is a true martingale, and $\widehat R_t^\delta$ is of order $\delta^{H\wedge \alpha}$. Therefore we obtain
\begin{align}
 \Vzp_t= \vz(t,X_t^\pi, \Zh{0}) + \EE\left[\left. \widehat N_T^\delta - \widehat N_t^\delta \right\vert \MCF_t\right] + \MCO(\delta^{H \wedge \alpha}) < \vz(t,X_t^\pi, \Zh{0}) + \MCO(\delta^{H\wedge \alpha}).\label{eq_Vexpansionpzt}
\end{align}

Now comparing the approximation \eqref{eq_Vtexpansion} with \eqref{eq_Vexpansionpz} \eqref{eq_Vexpansionpzt}, we obtain the desired result in Theorem \ref{thm_main}.
\end{proof}

\section{Conclusion}\label{sec_conclusion}

In this paper, we have considered the portfolio allocation problem in the context of a slowly varying fractional stochastic environment driven by a fractional OU process with $H\in(0,1)$, and when the investor tries to maximize her terminal utility with, first, power utilities, and, then, in a general class of utility functions. 

In the power utility case, using a martingale distortion representation for the value process and the espsilon-martingale decomposition method, we are able to derive a first order asymptotic approximation for both the optimal portfolio value and the optimal strategy. The first order correction for the optimal portfolio value has  both random and deterministic parts as in the linear option pricing problem studied in \cite{GaSo:15}. However, the  approximate optimal strategy does not involve a random part and can be easily implemented. We also show that the zeroth order of the optimal strategy generates the portfolio value up to the first order. We observe that it is more crucial to include the first order correction  in the case of $H$ small ($\delta^H$ large), and this ($H$ small) has been observed in volatility data (see \cite{roughvol}). 

Finally, we extend our analysis to the case of general utilities where we can derive the first order asymptotic optimality within a specific subclass of strategies $\widetilde \MCA_t^\delta$, which is of the form $\pzt + \delta^\alpha\pot$, with $\pzt$ and $\pot$ being of feedback forms and $\alpha>0$.

The case of fast varying fractional stochastic environment with $H\in (\frac{1}{2},1)$ is the topic of the paper \cite{FoHu2:17}.

\appendix
\section{Technical Lemmas}\label{app_lemmas}
In this section, we present several lemmas which are used in Section~\ref{sec_apptofSV} and \ref{sec_optimality}.

\begin{lem}\label{lem_moments}
\begin{enumerate}[(i)]
\item\label{lem_Zh} The slowly varing fractional factor $\Zh{t}$ defined in \eqref{def_Zh} is a stationary Gaussian process with zero mean and variance
\begin{equation}
\EE\left[\left(\Zh{t}\right)^2\right] = \int_{-\infty}^t \left(\MCK^\delta(t-s)\right)^2 \ud s = \int_0^\infty \MCK^2(s) \ud s = \sigma_{ou}^2,
\end{equation}
where $\sigma_{ou}^2$ is given in \eqref{eq_Zhvar} and free of $\delta$.
Therefore $\Zh{t}$ has finite moments of any order, and for any $p\in \NN^+$, $\Zh{\cdot} \in L^p([0,T]\times \Omega)$ uniformly in $\delta$.

Any adapted process  that $\chi_t \in \left[\Zh{0}\wedge\Zh{t}, \Zh{0}\vee\Zh{t}\right]$ also satisfies that  $\chi_\cdot \in L^p([0,T]\times \Omega)$ uniformly in $\delta$.

\item\label{lem_Zhdiff} The difference $\Zh{t} - \Zh{0}$ is a Gaussian random variable with zero mean and variance 
\begin{equation}
\EE\left[\left(\Zh{t}-\Zh{0}\right)^2\right] = \sigma_H^2 (\delta t)^{2H} + o(\delta^{2H}),
\end{equation}
where $\sigma_H^2 = (\Gamma(2H+1)\sin(\pi H))^{-1}$. Consequently, the $k^{th}$ moment of $\Zh{t} - \Zh{0}$ is of order $\delta^{kH}$, uniformly in $t \in [0,T]$. Moreover, $\Zh{\cdot} - \Zh{0}$ is of order $\delta^H$ in $ L^p([0,T] \times \Omega)$ sense, for any $p\in \NN^+$.

\item\label{lem_phi}	The random correction $\phi_t^\delta$ defined in \eqref{def_phi} is a normal random variable of order $\delta^H$ with zero mean and variance
	\begin{align}
	\EE\left[\left(\phi_t^\delta\right)^2\right] &= \frac{\delta^{2H}T^{2+2H}}{\Gamma^2(H+\frac{3}{2})} \int_0^\infty \left[ \left(1-\frac{t}{T} + v\right)^{H+\half} - v^{H + \half} - (1-\frac{t}{T})(H+\half)(v-\frac{t}{T})_+^{H-\half}\right]^2 \ud v \nonumber\\
	& \hspace{10pt} + \MCO(\delta^{2H+1}),
	\end{align}
where the integral is uniformly bounded in $t\in[0,T]$. Therefore, the $L^p([0,T]\times \Omega)$ norm $\phi_\cdot^\delta$ is of order $\delta^H$, for any $p\in \NN^+$.
\item\label{lem_psi}	The process $\left(\psi^\delta_t\right)_{t\in[0,T]}$ defined in \eqref{def_psi} is a square-integrable martingale satisfying
\begin{equation}
\ud \psi_t^\delta = \int_0^{T-t} \MCK^\delta(s) \ud s \ud W_t^Z := \left(\delta^H \theta_{t,T} + \delta^{H+1}\widetilde\theta_{t,T}\right) \ud W_t^Z,
\end{equation}
with $\theta_{t,T}$ and $\widetilde \theta_{t,T}$ given by
\begin{align*}
&	\theta_{t,T}	= \frac{1}{\Gamma(H+\frac{3}{2})}(T-t)^{H+\half}, 
\quad	\widetilde{\theta}_{t,T} =  \frac{a}{\Gamma(H+\half)}\int_0^{T-t}\int_0^s (s-u)^{H-\half}e^{-a\delta u}\ud u \ud s \leq \frac{a(T-t)^{H+\frac{3}{2}}}{\Gamma(H+\frac{5}{2})},
\end{align*}
and uniformly bounded in $t \in [0,T]$ and $\delta$. Consequently, one has
\begin{equation}
	\ud \average{\psi,W^Z}_t = \left(\int_{0}^{T-t} \MCK^\delta(s)\ud s\right)\ud t \text{ and } \ud \average{\psi}_t = \left(\int_{0}^{T-t} \MCK^\delta(s)\ud s\right)^2 \ud t.
\end{equation}

\end{enumerate}
\end{lem}
\begin{proof}
All can be computed directly, and we refer to the statements in  \cite[Section~6, Appendix~A]{GaSo:15}.
\end{proof}

\begin{lem}\label{lem_RtT}
The term $R_{[t,T]}$ defined in \eqref{def_RtT} is of order $\delta^{2H}$ in $L^2$ sense.
	\end{lem}
\begin{proof}
Taylor expanding $e^x$ at $x = 0$ gives
\begin{equation*}
R_{[t,T]} = B_{[t,T]} + \frac{e^{\chi_{[t,T]}}}{2}(A_{[t,T]} + B_{[t,T]})^2, 
\end{equation*}
with $\chi_{[t,T]}$ being  the Lagrange remainder $ \chi_{[t,T]} \in [(A_{[t,T]} + B_{[t,T]}) \wedge 0,  (A_{[t,T]} + B_{[t,T]})  \vee 0]$. Then, it suffices to (a) compute the moments of $A_{[t,T]}$ and $B_{[t,T]}$; and (b) prove $e^{\chi_{[t,T]}} \in L^4(\Omega)$. 

To this end, we first claim that 
\begin{equation*}
\EE\left[A_{[t,T]}^p\right] \sim \MCO(\delta^{pH}), \quad \EE\left[B_{[t,T]}^p\right] \sim \MCO(\delta^{2pH}), \quad  \forall p \in \NN.
\end{equation*}
They follow by the assumptions on $\lambda(\cdot)$ and its derivatives, properties of $\Zh{t}$ and $\Zh{t} - \Zh{0}$ as stated in Lemma\ref{lem_moments}\eqref{lem_Zh}-\eqref{lem_Zhdiff}, and the inequality:
\begin{equation*}
\EE\abs{\int_0^T   g_s \ud W_s}^p \leq \left(\frac{p(p-1)}{2}\right)^{p/2} T^{\frac{p-2}{2}}\EE\int_0^T \abs{g_s}^p\ud s, \quad \forall  p \geq 2 \text{ and } g_s \in \MCF_s.
\end{equation*}
%To compute the fourth-moment of $e^{\chi_{[t,T]}}$, 

For part (b), we notice that $0 \leq \EE[e^{4\chi_{[t,T]}}] \leq  \EE[e^{4(A_{[t,T]} + B_{[t,T]})  \vee 0}] \leq \EE[e^{4A_{[t,T]} + 4B_{[t,T]}} + 1]$, then, it remains to show $e^{A_{[t,T]} + B_{[t,T]}} \in L^4$. From the derivation in Theorem \ref{thm_Vtpowerexpansion}, one deduces
{\small
\begin{align*}
A_{[t,T]} + B_{[t,T]} & = \frac{1-\gamma}{2q\gamma}\int_t^T \lambda^2(\Zh{s}) -\lambda^2(\Zh{0}) \ud s + \rho \left(\frac{1-\gamma}{\gamma}\right)\int_t^T \lambda(\Zh{s}) - \lambda(\Zh{0}) \ud W_s^Z \\
&\quad - \half \rho^2\left(\frac{1-\gamma}{\gamma}\right)^2\int_t^T \lambda^2(\Zh{s}) - \lambda^2(\Zh{0}) \ud s,
\end{align*}
}
and 
\begin{equation*}
e^{4A_{[t,T]} + 4B_{[t,T]}} = e^{\frac{4(1-\gamma)}{q\gamma}\int_t^T \lambda^2(\Zh{s})-\lambda^2(\Zh{0})\ud s + \rho^2\left(\frac{1-\gamma}{\gamma}\right)^2 \int_t^T 6\lambda^2(\Zh{s}) + 10\lambda^2(\Zh{0}) - 16\lambda(\Zh{s})\lambda(\Zh{0}) \ud s }\cdot\mc{E}_{[t,T]},
\end{equation*}
where $\mc{E}_{[t,T]}$ is given by
\begin{equation*}
\mc{E}_{[t,T]} = e^{4\rho \left(\frac{1-\gamma}{\gamma}\right)\int_t^T \lambda(\Zh{s}) - \lambda(\Zh{0}) \ud W_s^Z - 8\rho^2 \left(\frac{1-\gamma}{\gamma}\right)^2\int_t^T \left(\lambda(\Zh{s}) - \lambda(\Zh{0})\right)^2 \ud s}.
\end{equation*}
Then,  the fact that $e^{A_{[t,T]} + B_{[t,T]}} \in L^4$ follows by $\EE[\mc{E}_{[t,T]}] = 1$ (Novikov's condition) and the boundedness of $\lambda(\cdot)$.
\end{proof}

\begin{lem}\label{lem_mt} The processes $\left(M^{(i)}\right)_{t\in [0,T]}$, $ i = 1, 2, 3$ defined in \eqref{def_m1}, \eqref{def_m2} and \eqref{def_m3} are true martingales with respect to the filtration $\MCF_t$, so does $\left(M\right)_{t\in[0,T]}$.
\end{lem}
\begin{proof}
We prove this result by showing $\EE\left[\average{M^{(1)}}_T^{1/2}\right] < \infty$, which is equivalent to $\EE\left[\sup_{s \leq T}\abs{M^{(1)}_s}\right] < \infty$ by Burkholder--Davis--Gundy  inequality. This implies that $M^{(1)}$ is a martingle.

To this end, we first bound its quadratic variation
\begin{align*}
\ud \average{M^{(1)}}_t &= \lambda^2(\Zh{t})R^2(t,X_t^\pz; \lambda(\Zh{t}))\left(\vz_x(t,X_t^\pz, \Zh{0})\right)^2\ud t \\
& \leq \lambda^2(\Zh{t}) C^2 \left(X_t^\pz \vz_x(t,X_t^\pz, \Zh{0})\right)^2 \ud t \leq \lambda^2(\Zh{t}) C^2 \left(\vz(t,X_t^\pz, \Zh{0})\right)^2 \ud t
\end{align*}
by using the estimate $R(t,x;\lambda(z)) \leq Cx$  and the concavity of $\vz$,
and then deduce
\begin{align*}
\EE\left[\average{M^{(1)}}_T^{1/2}\right] &\leq C^2 \EE\left[\left(\int_0^T \lambda^2(\Zh{s})\left(\vz(s,X_s^\pz, \Zh{0})\right)^2 \ud s \right)^{1/2}\right] \\
&\leq C^2 \EE^{1/4}\left[\int_0^T \lambda^4(\Zh{s}) \ud s\right] \cdot \EE^{1/4} \left[\int_0^T \left(\vz(s,X_s^\pz, \Zh{0})\right)^4 \ud s \right] < \infty,
\end{align*}
where to conclude, we have used Assumption~\ref{assump_valuefunc}, and Lemma~\ref{lem_moments}\eqref{lem_Zh} about $\Zh{s}$. 

The proofs for $M^{(2)}$ and $M^{(3)}$ are obtained in a similar way with estimates from \cite[Proposition~3.5]{FoHu:16}, which is of the form
\begin{equation}\label{eq_prop_Rbounds}
\abs{R^j(t,x;\lambda(z))\left(\partial_x^{(j+1)}R(t,x;\lambda(z))\right)} \leq K_j, \quad \forall 0 \leq j \leq 3, \quad \forall (t,x,z) \in [0,T)t\times \RR^+ \times \RR,
\end{equation}
 and Lemma~\ref{lem_moments}\eqref{lem_phi}-\eqref{lem_psi},
and thus we omit the details here.
\end{proof}

\begin{lem} The process $\left(R_t^\delta\right)_{t\in[0,T]}$ defined in \eqref{def_Rt} is of order $\delta^{2H}$.
\end{lem}
\begin{proof}
We shall prove that each term in $R_t^\delta$ is of order $\delta^{2H}$. The first term we deal with is $g_t^{(1)}\vz_x$ with $g_t^{(1)}$ defined in \eqref{def_g1}:
\begin{align*}
\abs{g_t^{(1)} \vz_x(t,X_t^\pz, \Zh{0})} &= \half\left(\Zh{t}-\Zh{0}\right)^2 \abs{2\left(\lambda'\right)^2 R + 2\lambda\lambda'' R + 4\lambda\lambda'R_z + \lambda^2 R_{zz}}_{z= \chi_t^{(1)}} \vz_x(t,X_t^\pz, \Zh{0}) \\
& \leq \half \left(\Zh{t}-\Zh{0}\right)^2 d(\chi_t^{(1)}) R(t,X_t^\pz; \lambda(\chi_t^{(1)})) \vz_x(t,X_t^\pz, \Zh{0})\\
& \leq \half \left(\Zh{t}-\Zh{0}\right)^2 d(\chi_t^{(1)}) CX_t^\pz \vz_x(t,X_t^\pz, \Zh{0}) \\
& \leq C \left(\Zh{t}-\Zh{0}\right)^2 d(\chi_t^{(1)}) \vz(t,X_t^\pz, \Zh{0}).
\end{align*}
Here the first inequality follows from \cite[Propositon~3.7]{FoHu:16}: there exists non-negative functions $\widetilde d_{01}(z)$ and $\widetilde d_{02}(z)$ that have mostly polynomial growth and satisfy
\begin{equation*}
\abs{R_z(t,x;\lambda(z))} \leq \widetilde d_{01}(z)R(t,x;\lambda(z)), \quad \abs{R_{zz}(t,x;\lambda(z))} \leq \widetilde d_{02}(z)R(t,x;\lambda(z)),
\end{equation*}
and thus $d(z)$ is also at most polynomially growing defined as
\begin{equation}
d(z) = \abs{2\left(\lambda'(z)\right)^2 + 2\lambda(z)\lambda''(z) + 4\lambda(z)\lambda'(z)\widetilde d_{01}(z) + \lambda^2(z)\widetilde d_{02}(z)}.
\end{equation}
The second inequality is given by the estimate $R(t,x;\lambda(z)) \leq Cx$ and the concavity of $\vz$.  
Therefore 
\begin{align*}
\EE\left[\int_0^T g_s^{(1)}\vz_x(s,X_s^\pz, \Zh{0}) \ud s \right] &\leq C\EE\left[\int_0^T\left(\Zh{s}-\Zh{0}\right)^2 d(\chi_s^{(1)}) \vz(s, X_s^\pz, \Zh{0}) \ud s\right]\\
&\hspace{-3cm} \leq \left[\EE\int_0^T \left(\Zh{s}-\Zh{0}\right)^8 \ud s \right]^{\frac{1}{4}}\left[\EE\int_0^T d^4(\chi_s^{(1)})\ud s \right]^{\frac{1}{4}}\left[\EE\int_0^T  \left(\vz(s,X_s^\pz, \Zh{0})\right)^2\ud s \right]^{\frac{1}{2}}
\end{align*}
and is of order $\delta^{2H}$. This is because, one has proved in Lemma~\ref{lem_moments}\eqref{lem_Zhdiff} that the first expectation is of order $\delta^{2H}$,
the second expectation is uniformly bounded in $\delta$ due to the polynomial growth property of $d(\cdot)$ and Lemma~\ref{lem_moments}\eqref{lem_Zh}, while the third term is uniformly bounded by Assumption~\ref{assump_valuefunc}\eqref{assump_vz}. 

Other terms contained in $R_t^\delta$ can be proved of order $\delta^{2H}$ in a similar way with additional Assumption~\ref{assump_valuefunc}\eqref{assump_vxx}, estimates \eqref{eq_prop_Rbounds}, Lemma~\ref{lem_moments}\eqref{lem_phi}-\eqref{lem_psi} and estimates from \cite[Proposition~4]{KaZa:16}.

\end{proof}

\section{Assumptions in Section \ref{sec_optimalitypz}}\label{appendix_addasump}

This set of assumptions is used in establishing the approximation accuracy \eqref{eq_Vexpansionpz} (resp. \eqref{eq_Vexpansionpzt}) to $V_t^\pi$ defined in \eqref{def_Vzpi}, namely, these assumptions will ensure that $\widetilde M_t^\delta$ (resp. $\widehat M_t^\delta$) is a true martingale and that $\widetilde R_t^{\delta}$ (resp. $\widehat R_t^\delta$)  is of order $\delta^{H + H \wedge \alpha}$ (resp. $\delta^{H \wedge \alpha}$).
\begin{assump}\label{assump_optimality}
	Let $\MCA_0(t,x,z)\left[\pzt,\pot,\alpha\right]$ be the family of trading strategies defined in \eqref{def_Atilde}. Recall that $X^\pi$ is the wealth generated by the strategy $\pi=\pzt+\delta^\alpha\pot$ as defined in \eqref{def_Xtilde}. In order to condense the notation,  we systematically omit  the argument $(s,X_s^\pi,\Zh{0})$ of $\vz$ and $\vo$, the argument $\Zh{s}$ of $\mu$ and $\sigma$, the argument $\Zh{0}$ of $\lambda$, and $(s, X_s^\pi, \Zh{s})$ of $\pzt$ and $\pot$ in what follows. According to the different cases, we further require:
	\begin{enumerate}[(i)]
		\item\label{assump_optimality_eq} If $\pzt \equiv \pz$, the following quantities are uniformly bounded in $\delta$:
		
		$\EE\int_0^T \left( (\mu\pot)_z\vert_{z = \widetilde\chi_s^{(1)}}\vz_x\right)^2\ud s$,
		$\EE\int_0^T \left( (\mu R\pot)_z\vert_{z = \widetilde\chi_s^{(2)}}\vz_{xx}\right)^2\ud s$,
		$\EE\int_0^T \left( \mu\pot \vz_x\right)^2\ud s$, 
		$\EE\int_0^T \left(\sigma\pot\vz_x\right)^2 \ud s $, 
		
%		$\EE\int_0^T \left( \mu R \pot \partial_{xx} D_1\vz\right)^2\ud s$,
		$\EE\int_0^T \left( \sigma^2\left(\pot\right)^2\partial_{xx} D_1\vz\right)^2\ud s$,		
		$\EE \left[\lambda^2\lambda'\int_0^T  \mu\pot \vo_x \ud s\right]$,		
%		$\EE \left[\lambda^2\lambda'\int_0^T  \mu R \pot \vo_{xx} \ud s\right]$,
		$\EE \left[\lambda^2\lambda'\int_0^T  \sigma^2\left(\pot\right)^2 \vo_{xx} \ud s\right]$,
		
		$\EE\left[\lambda\lambda'\left(\int_0^T \left(\sigma\pot\vz_x\phi_s^\delta\right)^2 \ud s\right)^\half\right]$, 
		$\EE\left[\lambda^2\lambda'\left(\int_0^T \left(\sigma\pot\vo_x\right)^2 \ud s\right)^\half\right]$,
		
		\item\label{assump_optimality_neq} If $\pzt \not\equiv \pz$, we require the uniformly boundedness (in $\delta$) of the following:

		$\EE\int_0^T \left( (\mu\pzt)_z\vert_{z = \widehat\chi_s^{(1)}}\vz_x\right)^2\ud s$,
		$\EE\int_0^T \left( (\sigma^2\left(\pzt\right)^2)_z\vert_{z = \widehat\chi_s^{(2)}}\vz_{xx}\right)^2\ud s$,
		$\EE\int_0^T  \mu\pot \vz_x\ud s$,
		%$\EE\int_0^T  \sigma^2\pzt\pot \vz_{xx}\ud s$,

		$\EE\int_0^T  \sigma^2 \left(\pot\right)^2\vz_{xx}\ud s$,
		$\EE\left(\int_0^T \left(\sigma\pzt\vz_x\right)^2 \ud s\right)^\half$,
		$\EE\left(\int_0^T \left(\sigma\pot\vz_x\right)^2 \ud s\right)^\half$.
	\end{enumerate}
\end{assump}

\bibliographystyle{plainnat}
\bibliography{Reference}

\end{document}